\pdfoutput=1
\documentclass[11pt]{article}
\usepackage[T1]{fontenc}
\usepackage{lmodern}
\usepackage{subfig}
\usepackage{amssymb,amsmath,amsthm,amsfonts}
\usepackage{graphics,graphicx,epsfig}
\usepackage{algorithmic}
\usepackage{algorithm}
\usepackage[margin=1 in]{geometry}
\usepackage{mathtools}
\usepackage{cuted}
\usepackage{hhline}
\usepackage{multirow}
\usepackage{array}
\usepackage{tabularx}
\usepackage{verbatim}
\usepackage{xcolor}
\usepackage{cite}
\usepackage{url}

\newtheorem{lem}{Lemma}[section]

\newtheorem{coro}[lem]{Corollary}
\newtheorem{thm}[lem]{Theorem}

\newtheorem{prop}[lem]{Proposition}

\newcommand{\cmark}{\ensuremath{\checkmark}}
\newcommand{\xmark}{\ensuremath{\times}}
\newcommand{\Fin}[1]{\mathcal{F}^{\mathrm{in}}_{#1}}
\newcommand{\Fout}[1]{\mathcal{F}^{\mathrm{out}}_{#1}}
\newcommand{\TLP}{T_{\mathrm{LP}}}
\newcommand{\OPT}{\mathrm{OPT}}
\title{Non-Splitting Coflow Scheduling with Provable Guarantees in Heterogeneous Parallel Networks\thanks{This work has been submitted to the IEEE for possible publication. Copyright may be transferred without notice, after which this version may no longer be accessible.}}
\author{Chi-Yeh~Chen 
\\ Department of Computer Science and Information
Engineering, \\ National Cheng Kung University, \\
Taiwan, ROC. \\
chency@mail.csie.ncku.edu.tw.}

\begin{document}
\maketitle
\begin{abstract}
A coflow is a collection of parallel flows whose completion time is determined by its last finishing flow. Modern data centers increasingly operate several switching fabrics in parallel, combining Electronic Packet Switches (EPS) with Optical Circuit Switches (OCS) whose circuits must be reconfigured either asynchronously (not-all-stop) or synchronously (all-stop). Existing approximation algorithms for hybrid EPS/OCS fabrics consider two network cores and split flows across them, while existing non-splitting algorithms for multiple cores consider cores of a single type. This paper studies non-splitting coflow scheduling on $m \ge 2$ heterogeneous network cores of all three types, where every flow is transmitted entirely through one core and a circuit, once established, may carry several flows. We formulate a linear program (LP) that charges reconfiguration delays per circuit, prove that it is a relaxation, and compute its optimal makespan bound exactly in polynomial time. A heavy-core rounding then assigns flows to cores, and established single-core schedulers order the traffic on each core. Let $\tau$ be the maximum number of flows incident to a port. For the makespan, we prove upper bounds of $\min\{\tau, m\}$, $2\min\{\tau, m\}$, and $2\min\{2\tau-1, m+\tau-1, 2m\}$ for EPS, not-all-stop OCS, and all-stop OCS cores, respectively; a hybrid network inherits the bound of its weakest architecture, which is at most $4m$. For the total weighted coflow completion time, an interval-indexed LP and an explicit group schedule that waits for release times yield $8R$ for zero release times and $8R+4$ for arbitrary release times, where $R$ is the corresponding per-group factor (at most $8m$). We also show that assignment heuristics that ignore reconfiguration delays have unbounded approximation ratios, and that our algorithms can be combined with any heuristic without losing their guarantees, while the LP bound certifies the quality of every returned schedule. Simulations show that the LP-based schemes reduce the makespan and the weighted completion time substantially compared with schemes that ignore reconfiguration delays, perform on par with reconfiguration-aware heuristics on average and better on adversarially found instances, and that the combined schemes are the best non-splitting schemes in all experiments.
\begin{keywords}
Scheduling algorithms, approximation algorithms, coflow, datacenter network, heterogeneous parallel network.
\end{keywords}
\end{abstract}

\section{Introduction}\label{sec:introduction}
Application-level computation and communication patterns in modern data centers are increasingly modeled using the coflow abstraction. Distributed frameworks characterized by structured traffic patterns, such as MapReduce~\cite{Dean2008}, Hadoop~\cite{Shvachko2010, borthakur2007hadoop}, Dryad~\cite{isard2007dryad}, and Spark~\cite{zaharia2010spark}, derive substantial performance gains from application-aware network scheduling~\cite{Chowdhury2014, Chowdhury2015, Zhang2016, Agarwal2018}. Typically, the computation stage of these frameworks generates intermediate data that must be transferred across distributed machines during the subsequent communication stage. A coflow abstracts this collection of parallel flows between two sets of machines, and the completion time of a coflow is determined by its last finishing flow~\cite{Chowdhury2012, shafiee2018improved}. Two objectives are standard: the \emph{makespan}, i.e., the time at which all coflows of a batch are completed, and the \emph{total weighted coflow completion time}, which favors important coflows.

\textbf{Background: switching fabrics.} An Electronic Packet Switch (EPS) forwards individual packets. Many flows can share an input or output port of an EPS simultaneously, so the only constraint is that the total rate of a port does not exceed its link capacity. An Optical Circuit Switch (OCS)~\cite{Zhang2021, Tan2021, Li2022} instead establishes \emph{circuits}, i.e., dedicated light paths between an input port and an output port. At any time the circuits of an OCS form a matching: every port is connected to at most one other port, and data can be sent only between connected ports, at the full line rate. Changing a circuit takes a \emph{reconfiguration delay} $\delta$, during which the affected ports cannot transmit~\cite{Zhang2021}. OCSs offer high rates and low power, but every new circuit costs $\delta$. Two reconfiguration models are studied in the literature~\cite{huang2016, Tan2021}. In the \emph{not-all-stop} (asynchronous) model, reconfiguring a circuit interrupts only the two ports of that circuit, while all other circuits keep transmitting. In the \emph{all-stop} (synchronous) model, the whole switch changes from one configuration (matching) to the next at once, and every circuit stops during each reconfiguration; a schedule is then a sequence of configurations, and the reconfiguration delays of all ports add up. Section~\ref{sec:Preliminaries} gives the formal definitions and their mathematical consequences.

\textbf{Heterogeneous parallel networks.} Large data centers increasingly operate several switching fabrics in parallel, e.g., fabrics of different generations~\cite{Huang2020}, hybrid electrical/optical designs~\cite{Liu2015, Li2022, 10411848}, and production systems that use OCSs to reconfigure their topology~\cite{poutievski2022jupiter, jouppi2023tpuv4}. Following~\cite{Huang2020, CHEN2023104752}, we call each fabric a \emph{network core}; every source and destination is connected to every core, and the cores may differ in type, rate, and reconfiguration delay. A scheduler must decide, for each flow, which core carries it, and must then schedule the traffic of each core.

\textbf{Splitting versus non-splitting.} If a flow is split over several cores, its packets traverse fabrics with different rates and reconfiguration schedules, so they arrive out of order. Receivers must then buffer and resequence data, and transport protocols react poorly to reordering~\cite{kandula2007flare, he2015presto}. For this reason, many systems keep each flow on a single path~\cite{kandula2007flare} and several schedulers for parallel networks assign every flow to one core~\cite{Huang2020, CHEN2023104752}. Splitting, however, also has benefits: it balances load at a finer granularity. In this paper we adopt the non-splitting model, in which every flow is transmitted entirely through one core while different flows of the same coflow may use different cores, and we quantify the cost of this restriction experimentally in Section~\ref{sec:Results}.

\textbf{Scope.} The coflow abstraction applies to communication stages that complete when their last flow completes, such as the shuffle of data-parallel frameworks or a single bulk-synchronous exchange (e.g., an all-to-all step) of distributed training. Multi-stage collectives with dependencies between steps, or applications that consume flows at different times~\cite{pan2022echelon}, are outside the scope of this paper. Moreover, a fabric that contains all three core types (EPS, not-all-stop OCS, and all-stop OCS) at the same time is mainly a theoretical generalization: our results specialize to the practically relevant combinations, such as EPS with one type of OCS, or several OCS planes of the same type.

\textbf{Approximation guarantees.} An algorithm is a $\rho$-approximation if, on every instance, the objective value of its schedule is at most $\rho$ times the optimal value $\OPT$. All guarantees in this paper are proven upper bounds on this ratio; we do not claim that they are tight. Both problems are $\mathcal{NP}$-hard: the weighted problem already for a single core~\cite{Qiu2015}, and the makespan problem already for two identical EPS cores and a single port pair, by a reduction from \textsc{Partition}. Hence the guarantees are obtained by comparing the schedule with a lower bound on $\OPT$ given by a linear program (LP).

\textbf{Technical challenge.} Non-splitting assignment on heterogeneous cores resembles scheduling on unrelated machines~\cite{lenstra1990approximation, shmoys1993approximation}, but every flow consumes capacity at two ports of its core, and on an OCS core the cost of a circuit is paid once per port pair and core, not once per flow. Moreover, the guarantees of single-core OCS schedulers~\cite{huang2016, Tan2021} are stated with respect to quantities, such as the number of circuits at a port, that depend on the assignment and are not linear in it. Our approach is to (i) use an LP whose setup variables charge one reconfiguration delay per (port pair, core), which keeps the LP a valid relaxation when circuits are reused; (ii) assign each flow only to a core that carries at least a $1/m$ fraction of it in the LP, which simultaneously controls the data volume and the number of circuits at \emph{both} endpoints of every flow; and (iii) combine two different lower bounds to analyze the all-stop model, whose reconfigurations couple all ports of a core.

\subsection{Our Contributions}
This paper addresses non-splitting coflow scheduling in heterogeneous parallel networks with $m \ge 2$ cores consisting of EPS, not-all-stop OCS, and all-stop OCS. For EPS-only parallel networks, non-splitting assignment has already been studied~\cite{Huang2020, CHEN2023104752}; previous results with guarantees for hybrid EPS/OCS fabrics~\cite{10411848, 11200997} consider one EPS and one OCS and split flows between them. The new results of this paper are guarantees for $m \ge 2$ cores that may include all-stop OCS, not-all-stop OCS, and EPS at the same time, without splitting flows. Specifically:
\begin{itemize}
\item \textbf{Model and lower bound.} We formalize the three core types, allow a circuit to carry several flows of the same port pair, and give an LP relaxation that charges reconfiguration delays per circuit. We prove that it is a relaxation and compute its optimal makespan bound exactly in polynomial time, also when processing times are fractional (Section~\ref{sec:Algorithm3}).
\item \textbf{Makespan.} A heavy-core rounding (LP-Max) and its greedy refinement (LP-Greedy) assign flows to cores; established single-core schedulers order the traffic. Let $\tau$ be the maximum number of flows incident to a port. We prove the guarantees of Table~\ref{tab:guarantees}: $\min\{\tau, m\}$ for EPS, $2\min\{\tau, m\}$ for not-all-stop OCS, and $2\min\{2\tau-1, m+\tau-1, 2m\}$ for all-stop OCS; a hybrid network inherits the bound of its weakest architecture, which is at most $4m$ for every $\tau$.
\item \textbf{Total weighted coflow completion time.} An interval-indexed LP, proven to be a relaxation, guides a grouping of flows; the groups are executed with explicit start times that wait for release times. The guarantees are $8R$ for zero release times and $8R+4$ for arbitrary release times, where $R$ is the per-group factor of Table~\ref{tab:guarantees} (Section~\ref{sec:Algorithm5}).
\item \textbf{Guarantees in practice.} We show that heuristics that balance data volumes, such as greedy rules and the assignment of Weaver~\cite{Huang2020}, have unbounded approximation ratios in OCS-based networks, whereas our guarantees hold for every instance and are kept when our schedule is combined with any heuristic; the LP bound also certifies the quality of the returned schedule on every instance (Section~\ref{sec:robust}).
\item \textbf{Evaluation.} We evaluate both objectives, all rounding variants, adaptations of existing schemes, and reconfiguration-aware heuristics; we compare with exact optima of small instances, with a lower bound for all splitting-based schedules, and on instances found by adversarial search, and report the running time and memory of every phase (Section~\ref{sec:Results}).
\end{itemize}

\begin{table}[!t]
\caption{Proven upper bounds on the approximation ratio ($m$: number of cores; $\tau$: maximum number of flows incident to a port). A hybrid network inherits the largest entry among the architectures it contains. For the weighted completion time (WCT), the first value holds for zero release times and the second for arbitrary release times.}
\label{tab:guarantees}
\centering
\footnotesize
\renewcommand{\arraystretch}{1.3}
\setlength{\tabcolsep}{4pt}
\begin{tabular}{|l|c|c|}
\hline
Core type & Makespan & WCT \\ \hline
EPS & $\min\{\tau, m\}$ & $8R_{E}$ / $8R_{E}+4$ \\ \hline
Not-all-stop OCS & $2\min\{\tau, m\}$ & $8R_{N}$ / $8R_{N}+4$ \\ \hline
All-stop OCS & $2\min\{2\tau-1, m+\tau-1, 2m\}$ & $8R_{A}$ / $8R_{A}+4$ \\ \hline
\multicolumn{3}{|l|}{$R_{E}=\min\{\tau,2m\}$, \ $R_{N}=2\min\{\tau,2m\}$,} \\
\multicolumn{3}{|l|}{$R_{A}=2\min\{2\tau-1,2m+\tau-1,4m\}$} \\ \hline
\end{tabular}
\end{table}

\subsection{Organization}
Section~\ref{sec:Related} reviews related work. Section~\ref{sec:Preliminaries} formulates the system model. Section~\ref{sec:Algorithm3} presents the LP-based framework for the makespan and its analysis, and Section~\ref{sec:Algorithm5} extends it to the total weighted coflow completion time and discusses the guarantees for both objectives. Section~\ref{sec:Results} evaluates the algorithms, and Section~\ref{sec:Conclusion} concludes the paper. The appendices contain the Birkhoff-von Neumann procedure, implementation details of the single-core schedulers, the matching-based rounding, and additional experiments.

\section{Related Work}\label{sec:Related}
Chowdhury and Stoica~\cite{Chowdhury2012} introduced the coflow abstraction. Table~\ref{tab:network_comparison} compares the closely related works; it is adapted and expanded from~\cite{10411848, 11200997}.

\begin{table*}[htbp]
\centering
\caption{Comparison of Related Works (Splitting: whether the traffic of a flow may be divided among several cores; ``--'': single-core works or works that do not address the assignment of individual flows).}
\label{tab:network_comparison}
\small\setlength{\tabcolsep}{4pt}
\begin{tabular}{|>{\raggedright\arraybackslash}p{1.9in}|c|c|c|c|c|c|}
\hline
\multicolumn{1}{|c|}{\multirow{2}{*}{Works}} & \multirow{2}{*}{\shortstack{EPS-\\enable}} & \multicolumn{2}{c|}{OCS-enable} & No. of & \multirow{2}{*}{Splitting} & \multirow{2}{*}{\shortstack{Performance\\Guarantee}} \\ \cline{3-4}
\multicolumn{1}{|c|}{} & & \textit{all-stop} & \textit{not-all-stop} & Network Cores & & \\ \hline
\cite{chowdhury2011managing, Chowdhury2014, Chowdhury2015, Zhang2016, zhao2015rapier} & \cmark & \xmark & \xmark & 1 & -- & \xmark \\ \hline
\cite{Qiu2015, khuller2016brief, shafiee2018improved, ahmadi2020scheduling, Li2016, Agarwal2018, fukunaga2022integrality, rohwedder2025coflow} & \cmark & \xmark & \xmark & 1 & -- & \cmark \\ \hline
\cite{Xu2018} & \xmark & \cmark & \xmark & 1 & -- & \xmark \\ \hline
\cite{Zhang2019, Tan2021} & \xmark & \cmark & \xmark & 1 & -- & \cmark \\ \hline
\cite{huang2016, Zhang2021} & \xmark & \xmark & \cmark & 1 & -- & \cmark \\ \hline
\cite{Li2022, 10217994} & \cmark & \xmark & \cmark & 2 & -- & \xmark \\ \hline
\cite{Liu2015} & \cmark & \cmark & \xmark & 2 & \cmark & \xmark \\ \hline
\cite{10411848} & \cmark & \cmark & \xmark & 2 & \cmark & \cmark \\ \hline
\cite{11200997} & \cmark & \cmark & \cmark & 2 & \cmark & \cmark \\ \hline
\cite{CHEN2023104752,chen2023efficient1, chen2023efficient2,Huang2020} & \cmark  & \xmark & \xmark & $m$ & \xmark & \cmark \\ \hline
\textbf{Our Work} & \cmark & \cmark & \cmark & $m$ & \xmark & \cmark \\ \hline
\end{tabular}
\end{table*}

\subsection{Coflow Scheduling in EPS-based Networks}
For a single EPS core, Qiu \textit{et al.}~\cite{Qiu2015} gave the first constant-factor approximation algorithms for minimizing the total weighted completion time, with ratios $\frac{64}{3}$ for zero release times and $\frac{76}{3}$ for arbitrary release times (the latter as corrected in~\cite{ahmadi2020scheduling}). Khuller and Purohit~\cite{khuller2016brief} improved them to $8$ and $12$. Ratios of $4$ for zero release times were obtained in~\cite{shafiee2018improved, ahmadi2020scheduling, Agarwal2018, fukunaga2022integrality}, and ratios of $5$ for arbitrary release times in~\cite{shafiee2018improved, ahmadi2020scheduling}; Fukunaga~\cite{fukunaga2022integrality} also analyzed the integrality gap of the time-indexed LP relaxation. Rohwedder and Schnaars~\cite{rohwedder2025coflow} recently improved these ratios to $3.415$ and $4.36$ by iterated rounding.

For parallel networks consisting of EPS cores, Huang \textit{et al.}~\cite{Huang2020} designed Weaver, which assigns every flow of a coflow to one core and whose coflow completion time is within a factor $\sum_q r_q/\max_q r_q \le m$ of the splittable optimum. Chen~\cite{CHEN2023104752} gave an $O\left(\frac{\log m}{\log \log m}\right)$-approximation for flow-level (non-splitting) scheduling in heterogeneous parallel networks with arbitrary release times. For identical parallel networks, Chen~\cite{chen2023efficient1} obtained $6-\frac{2}{m}$ and $5-\frac{2}{m}$ for arbitrary and zero release times, and Chen~\cite{chen2023efficient2} presented an $O(\chi)$-approximation for coflows with precedence constraints, where $\chi$ is the number of coflows on the longest path of the precedence graph. These works consider EPS cores only, i.e., no reconfiguration delays; for EPS-only networks, our bounds do not improve upon~\cite{CHEN2023104752}.

\subsection{Coflow Scheduling in OCS-based Networks}
Under the all-stop model, Tan \textit{et al.}~\cite{Tan2021} proposed Reco-Sin, a $2$-approximation for a single coflow, and Reco-Mul for multiple coflows (a preliminary version appeared in~\cite{Zhang2019}); OMCO~\cite{Xu2018} is an online heuristic. Under the not-all-stop model, Huang \textit{et al.}~\cite{huang2016} proposed Sunflow, which has a constant-factor guarantee for a single coflow, and Zhang \textit{et al.}~\cite{Zhang2021} gave a $4$-approximation for multiple coflows.

\subsection{Coflow Scheduling in Hybrid EPS/OCS Networks}
Solstice~\cite{Liu2015}, Co-scheduler~\cite{Li2022}, and the scheme of Jiang \textit{et al.}~\cite{10217994} are heuristics for one EPS and one OCS. Wang \textit{et al.}~\cite{10411848} proposed a $2\tau'$-approximation for a single coflow in a hybrid switch with one EPS and one all-stop OCS, where $\tau'$ is the maximum number of nonzero entries in a row or column of the demand matrix, and a $\mu\tau'_{\max}$-approximation for multiple coflows, where $\mu = 4M w_{\max}/w_{\min}$ depends on the number $M$ of coflows and on the weights. In~\cite{11200997}, they showed how to convert any $\lambda$-approximate pure-OCS circuit schedule into a $(\lambda+1)$-approximate hybrid schedule by an optimal partitioning of the traffic demand, for both reconfiguration models. The essential differences from our work are the following. \emph{Formulation:} \cite{10411848, 11200997} decide how much of each demand is sent through the EPS and how much through the OCS, i.e., the decision variables are continuous partitions of every flow; we decide a discrete assignment of every flow to one of $m$ cores. \emph{Constraints:} \cite{10411848, 11200997} consider two cores (one EPS and one OCS), whereas we consider $m$ cores of three types at once together with integrality of the assignment. \emph{Proof techniques:} \cite{10411848, 11200997} start from a circuit schedule of the OCS and let the EPS absorb traffic, e.g., during reconfigurations; we compare an integral assignment with a fractional LP solution through a rounding argument that controls data and circuits at both endpoints of every flow, and then compose it with the bounds of single-core schedulers.

\section{System Model and Problem Formulation}\label{sec:Preliminaries}
We consider a set of coflows $\mathcal{K}$ and a set $\mathcal{M}$ of $m$ heterogeneous network cores, aiming to minimize the makespan and the total weighted coflow completion time. Each core is an $N \times N$ non-blocking switch. The set $\mathcal{M}$ is partitioned into three disjoint subsets: $\mathcal{M} = \mathcal{M}_{e} \cup \mathcal{M}_{o}^{1} \cup \mathcal{M}_{o}^{2}$, such that $m = m_{e} + m_{o}^{1} + m_{o}^{2}$. Here, $\mathcal{M}_{e}$, $\mathcal{M}_{o}^{1}$, and $\mathcal{M}_{o}^{2}$ denote the $m_{e}$ EPS cores, the $m_{o}^{1}$ all-stop OCS cores, and the $m_{o}^{2}$ not-all-stop OCS cores, respectively, and $\mathcal{M}_{o}=\mathcal{M}_{o}^{1} \cup \mathcal{M}_{o}^{2}$ is the set of OCS cores. Each core has $N$ input links connected to $N$ source servers and $N$ output links connected to $N$ destination servers. The notation is summarized in Table~\ref{tab:notations}.

\begin{table}[!t]
\caption{Notation and Terminology}
    \centering
        \begin{tabular}{||c|p{4.6in}||}
    \hline
     $N$      & The number of input/output ports.         \\
    \hline
		$\mathcal{M}$, $m$  & The set and the number of network cores.\\
    \hline
		$\mathcal{M}_{e}$, $\mathcal{M}_{o}^{1}$, $\mathcal{M}_{o}^{2}$ & EPS, all-stop OCS, and not-all-stop OCS cores; $\mathcal{M}_{o}=\mathcal{M}_{o}^{1}\cup\mathcal{M}_{o}^{2}$.\\
    \hline
     $\mathcal{I}, \mathcal{J}$ & The source server set and the destination server set.         \\
    \hline
     $\mathcal{K}$ & The set of coflows.         \\
    \hline
     $n$      & The number of flows.         \\
    \hline
     $D^{(k)}$     & The demand matrix of coflow $k$. \\
    \hline
     $d_{ijk}$     & The size of the flow to be transferred from input $i$ to output $j$ in coflow $k$.   \\
    \hline
		 $\mathcal{F}$ & The set of all flows. \\
    \hline
		 $\Fin{i}$, $\Fout{j}$ & The flows with source $i$; the flows with destination $j$. \\
		\hline
		 $r_{q}$, $\delta_{q}$ & The link rate and the reconfiguration delay of core $q$ ($\delta_q=0$ for $q\in\mathcal{M}_e$). \\
		\hline
$p_{ijkq}$	& $d_{ijk}/r_q+\delta_q$: the time needed by flow $(i, j, k)$ alone on core $q$. \\
		\hline
		$\tau$  & The maximum flow degree of a port, $\max\left\{\max_{i} |\Fin{i}|,\max_{j} |\Fout{j}|\right\}$. \\
    \hline
		 $C_k$, $C_{ijk}$     & The completion time of coflow $k$; of flow $(i,j,k)$.   \\
    \hline
     $a_k$, $w_{k}$     & The release time and the weight of coflow $k$.  \\
    \hline
		$\TLP$, $\OPT$    & The optimal makespan bound of LP($T$); the optimal objective value. \\
		\hline
		$t_{l}$, $L$ & $t_{l}=2^l$ for $l=0, 1, \ldots, L$; $L$ is the number of time intervals. \\
		\hline
        \end{tabular}
    \label{tab:notations}
\end{table}

\subsection{Network Model}\label{sec:network}
This paper considers a data center network whose $N$ ingress and $N$ egress ports (typically Top-of-Rack switches) are connected to all $m$ cores. At the ingress, traffic is buffered in Virtual Output Queues (VOQs), one per destination. Let $r_{q}$ denote the link rate (data per unit time) of core $q$, identical for all links of the core, and let $\mathcal{I} = \{1, \ldots, N\}$ and $\mathcal{J} = \{1, \ldots, N\}$ be the sets of source and destination servers; each core is thus a complete bipartite graph between $\mathcal{I}$ and $\mathcal{J}$. The three core types obey the following rules.
\begin{itemize}
\item \textbf{EPS core.} At any time, every flow on the core is sent at some rate, and the total rate through each input port and through each output port is at most $r_q$. Flows can be preempted and resumed at no cost.
\item \textbf{OCS core.} At any time, the established circuits form a matching between $\mathcal{I}$ and $\mathcal{J}$. Data from $i$ to $j$ can be sent, at rate $r_q$, only while circuit $(i,j)$ is established. Initially no circuit is established, and establishing a circuit takes $\delta_q$ time, during which the involved ports cannot transmit on the core. A circuit, once established, can carry any number of flows of its port pair, consecutively and without further delay (\emph{circuit reuse}).
\item \textbf{Not-all-stop OCS core.} Circuits are reconfigured individually: establishing circuit $(i,j)$ interrupts only ports $i$ and $j$ of the core during $\delta_q$, while all other circuits continue to transmit.
\item \textbf{All-stop OCS core.} The core alternates between reconfiguration periods of length $\delta_q$ and configuration periods. During a configuration period the set of circuits is fixed; during a reconfiguration period no port of the core transmits. Hence the schedule of the core is a sequence of matchings $\pi_1,\pi_2,\ldots$ held for durations $\ell_1,\ell_2,\ldots$, and its length is $\sum_{c}(\delta_q+\ell_c)$.
\end{itemize}

\textbf{Mathematical consequences.} For a fixed assignment of flows to core $q$, let $\rho_{iq}$ be the total data at port $i$ divided by $r_q$, and let $\deg_q(i)$ be the number of distinct partners of port $i$ among these flows. In both OCS models, port $i$ needs at least $\deg_q(i)$ circuit establishments, which are disjoint in time and disjoint from its transmissions, so every schedule of the core lasts at least $\rho_{iq}+\delta_q \deg_q(i)$. This per-port bound also holds for EPS cores with $\delta_q=0$, and it is the only bound for not-all-stop OCS. For all-stop OCS, reconfiguration periods are shared by all ports, and port $a$ with the largest data and port $b$ with the largest degree may differ: every schedule needs at least $\deg_q(b)$ configurations, all reconfiguration periods precede or separate configuration periods, and port $a$ transmits only during configuration periods. Consequently, an all-stop core needs at least $\rho_{aq} + \delta_q \deg_q(b)$ time, which can be almost twice the largest per-port bound. This global coupling is the reason why our bounds for all-stop cores are larger than those for not-all-stop cores.

\subsection{Traffic Model}
A coflow consists of a collection of independent data flows whose overall completion time is determined by the last finishing flow. Let $D^{(k)}=\left(d_{ijk}\right)_{i,j=1}^{N}$ represent the demand matrix for coflow $k$, where $d_{ijk}$ denotes the volume of data to be transmitted from input port $i$ to output port $j$. Each flow is identified by a tuple $(i, j, k)$ with $i \in \mathcal{I}$, $j \in \mathcal{J}$, and $k\in \mathcal{K}$. All flows of coflow $k$ are released at time $a_k$~\cite{Qiu2015}, and coflow $k$ has weight $w_k > 0$. Data sizes, rates, and delays are positive rational numbers; no integrality is assumed.

We treat flows as atomic entities that cannot be split across cores: all data of a flow $(i, j, k)$ is routed through the same core, although different flows of the same coflow may be routed through different cores. We define the flow sets
\begin{itemize}
\item $\Fin{i}=\left\{(i, j, k) \mid d_{ijk}>0, j \in \mathcal{J}, k \in \mathcal{K} \right\}$,
\item $\Fout{j}=\left\{(i, j, k) \mid d_{ijk}>0, i \in \mathcal{I}, k \in \mathcal{K} \right\}$,
\item $\mathcal{F}=\left\{(i, j, k) \mid d_{ijk}>0, i \in \mathcal{I}, j \in \mathcal{J}, k \in \mathcal{K} \right\}$,
\end{itemize}
and let $n=|\mathcal{F}|$ and $\tau = \max \{ \max_{i} |\Fin{i}|, \max_{j} |\Fout{j}| \}$, the maximum number of flows incident to a port. The time needed by flow $(i,j,k)$ on core $q$ when it is alone, including the establishment of its circuit, is
\begin{equation*}
p_{ijkq} = \frac{d_{ijk}}{r_{q}} + \delta_{q}, \quad \text{where } \delta_q=0 \text{ for } q\in\mathcal{M}_e.
\end{equation*}
Because of circuit reuse, $\delta_q$ is charged once per port pair and core in a schedule, not once per flow.

\subsection{Problem Formulation}
A schedule assigns every flow to a core and transmits all flows according to the rules of Section~\ref{sec:network}. Let $C_{ijk}$ denote the completion time of flow $(i,j,k)$ and $C_k=\max_{(i,j,k)\in\mathcal{F}} C_{ijk}$ that of coflow $k$. For the makespan, all coflows are released at time $0$ and we minimize $C_{\max}=\max_{k\in\mathcal{K}} C_{k}$. For the second objective, coflows have release times $a_k$, no data of coflow $k$ is sent before $a_k$, and we minimize $\sum_{k \in \mathcal{K}} w_{k}C_{k}$. $\OPT$ denotes the optimal value of the respective objective among all non-splitting schedules.

\section{Approximation Algorithms for Minimizing the Makespan}\label{sec:Algorithm3}
Our framework decouples the problem into two stages. In the first stage, an LP relaxation and a rounding procedure map each flow to one core. In the second stage, the flows assigned to each core are scheduled by an established single-core scheduler: Birkhoff-von Neumann decomposition (BvN)~\cite{marcus1959diagonals, Qiu2015} for EPS, Sunflow~\cite{huang2016} for not-all-stop OCS, and Reco-Sin~\cite{Tan2021} for all-stop OCS. The contribution of this section is the first stage and its analysis, which allow the guarantees of the single-core schedulers to be composed across heterogeneous cores.

\subsection{LP Relaxation}
For each flow $(i, j, k)\in \mathcal{F}$ and core $q \in \mathcal{M}$, the variable $x_{ijkq}$ indicates whether flow $(i, j, k)$ is routed through core $q$. For each port pair $(i,j)$ and OCS core $q\in\mathcal{M}_o$, the variable $y_{ijq}$ indicates whether circuit $(i,j)$ is used on core $q$. For a time bound $T$, consider the following LP:
\begin{subequations}\label{coflow:interval}
\begin{align}
& \text{LP($T$): find } (\textbf{x},\textbf{y}) \text{ such that} \tag{\ref{coflow:interval}}\\
& \sum_{q\in \mathcal{M}} x_{ijkq} = 1, \qquad\qquad\qquad\qquad \forall (i, j, k)\in \mathcal{F}, \label{interval:a} \\
& \sum_{(i, j, k) \in \Fin{i}} \tfrac{d_{ijk}}{r_q} x_{ijkq} + \delta_q\sum_{j\in\mathcal{J}} y_{ijq} \leq T, \quad \forall i, q, \label{interval:b} \\
& \sum_{(i, j, k) \in \Fout{j}} \tfrac{d_{ijk}}{r_q} x_{ijkq} + \delta_q\sum_{i\in\mathcal{I}} y_{ijq} \leq T, \quad \forall j, q, \label{interval:c} \\
& x_{ijkq} \leq y_{ijq}, \qquad\qquad\qquad\quad \forall (i, j, k)\in\mathcal{F}, q\in \mathcal{M}_o, \label{interval:y} \\
& x_{ijkq}, y_{ijq} \geq 0, \qquad\qquad\qquad\qquad\quad\ \ \forall (i, j, k), q, \label{interval:d} \\
& x_{ijkq} = 0 \ \text{ if } p_{ijkq}>T, \qquad\qquad\qquad\ \ \forall (i, j, k), q. \label{interval:e}
\end{align}
\end{subequations}
Here $i\in\mathcal{I}$, $j\in\mathcal{J}$, and $q\in\mathcal{M}$; for an EPS core, $\delta_q=0$ and no $y$ variable appears. Constraint (\ref{interval:a}) assigns each flow to the cores. Constraints (\ref{interval:b}) and (\ref{interval:c}) bound the busy time of every input and output port of every core: the transmission time of its data plus one reconfiguration delay per circuit. Constraint (\ref{interval:y}) states that a flow can use core $q$ only if its circuit is established there, so that a circuit is charged once no matter how many flows it carries. Constraint (\ref{interval:e}) excludes a core on which a flow alone would need more than $T$. Let $\TLP=\min\{T \ge 0 : \text{LP}(T) \text{ is feasible}\}$.

\begin{lem}\label{lem:relax}
For every feasible schedule with makespan $C$, $\text{LP}(C)$ is feasible. Hence $\TLP \le \OPT$.
\end{lem}
\begin{proof}
Take a schedule with makespan $C$. Set $x_{ijkq}=1$ if flow $(i,j,k)$ is sent through core $q$ and $0$ otherwise, and set $y_{ijq}=1$ if some flow of port pair $(i,j)$ is sent through OCS core $q$ and $0$ otherwise. Constraints (\ref{interval:a}), (\ref{interval:y}), and (\ref{interval:d}) hold by construction. Consider port $i$ of core $q$. If $q$ is an EPS core, port $i$ sends at total rate at most $r_q$, so its data needs at least $\sum_{(i,j,k)\in\Fin{i}} d_{ijk}x_{ijkq}/r_q \le C$ time. If $q$ is an OCS core, port $i$ sends only through established circuits, at rate $r_q$. For every $j$ with $y_{ijq}=1$, circuit $(i,j)$ is established at least once during $[0,C]$, since no circuit exists initially. These establishments take $\delta_q$ each; they are disjoint in time, because port $i$ is connected to at most one partner at a time, and port $i$ does not transmit during them in either reconfiguration model. Hence the busy time of port $i$, which is at most $C$, is at least the left-hand side of (\ref{interval:b}) with $T=C$. Output ports are symmetric. Finally, if $x_{ijkq}=1$, then the flow completes by $C$ and needs its circuit ($\delta_q$) and its transmission ($d_{ijk}/r_q$), so $p_{ijkq} \le C$ and (\ref{interval:e}) holds. Applying this to an optimal schedule gives $\TLP\le\OPT$.
\end{proof}

\subsection{Computing the LP Bound Exactly}
Constraint (\ref{interval:e}) depends on $T$, so $\TLP$ is not obtained from a single LP. Since the processing times are fractional in general, a search over integers is not valid; instead, Algorithm~\ref{Alg1} searches over the breakpoints of (\ref{interval:e}). Let $P_1 < P_2 < \cdots < P_s$ be the distinct values of $\{p_{ijkq}\}$, let $S_a=\{((i,j,k),q) : p_{ijkq}\le P_a\}$, and let $\Phi(S)$ be the optimal value of the LP that minimizes $T$ subject to (\ref{interval:a})--(\ref{interval:d}) and $x_{ijkq}=0$ for $((i,j,k),q)\notin S$, with $\Phi(S)=\infty$ if this LP is infeasible.

\begin{algorithm}
\caption{Computing $\TLP$ (Breakpoint Search)}
    \begin{algorithmic}[1]
				\REQUIRE flows $\mathcal{F}$, cores $\mathcal{M}$, values $p_{ijkq}$
				\ENSURE $\TLP$ and a feasible solution $(\textbf{x},\textbf{y})$ of LP($\TLP$)
				\STATE Sort the distinct values of $\{p_{ijkq}\}$ as $P_1 < \cdots < P_s$. \label{Alg1:1}
				\IF{$\Phi(S_s) \ge P_s$}
				    \STATE \textbf{return} $\TLP=\Phi(S_s)$ and an optimal solution of $\Phi(S_s)$. \label{Alg1:2}
				\ENDIF
				\STATE $l \leftarrow 1$, $u \leftarrow s$. \label{Alg1:5}
				\WHILE{$l < u$} \label{Alg1:6}
							 \STATE $a=\left\lfloor (l+u)/2\right\rfloor$.
							 \IF{$\Phi(S_a) \le P_a$}
							     \STATE $u=a$
							 \ELSE
							     \STATE $l=a+1$
							 \ENDIF \label{Alg1:7}
				\ENDWHILE
				\STATE $a^*\leftarrow u$.
				\IF{$a^*>1$ and $\Phi(S_{a^*-1}) < P_{a^*}$}
					\STATE \textbf{return} $\TLP=\Phi(S_{a^*-1})$ and an optimal solution of $\Phi(S_{a^*-1})$. \label{Alg1:9}
				\ELSE
					\STATE \textbf{return} $\TLP=P_{a^*}$ and an optimal solution of $\Phi(S_{a^*})$. \label{Alg1:8}
				\ENDIF
   \end{algorithmic}
\label{Alg1}
\end{algorithm}

\begin{lem}\label{lem:bsearch}
Algorithm~\ref{Alg1} returns $\TLP$ and a feasible solution of LP($\TLP$) after $O(\log (nm))$ LP solutions.
\end{lem}
\begin{proof}
Recall that $P_1<\cdots<P_s$ are the distinct processing times and that $S_a$ contains the variables with $p_{ijkq}\le P_a$. For $T \in [P_a, P_{a+1})$ (with $P_{s+1}=\infty$), constraint (\ref{interval:e}) allows exactly the variables in $S_a$, so LP($T$) is feasible if and only if $T \ge \Phi(S_a)$; for $T<P_1$, no variable is allowed and LP($T$) is infeasible. Since $S_a \subseteq S_{a+1}$, $\Phi(S_a)$ is non-increasing in $a$, while $P_a$ is increasing; hence, once the predicate $\Phi(S_a)\le P_a$ holds for some $a$, it holds for all larger $a$.

\emph{Case 1: $\Phi(S_s)\ge P_s$.} Every $T\ge P_s$ lies in the last interval, so LP($T$) is feasible if and only if $T\ge\Phi(S_s)$. Every $T<P_s$ is either smaller than $P_1$ or lies in an interval $[P_a,P_{a+1})$ with $a<s$; in the latter case, $\Phi(S_a)\ge\Phi(S_s)\ge P_s>T$, so LP($T$) is infeasible. Hence $\TLP=\Phi(S_s)$, as returned in line~\ref{Alg1:2}.

\emph{Case 2: $\Phi(S_s)< P_s$.} The predicate holds for $a=s$, and the binary search finds the smallest index $a^*$ for which it holds. Since $\Phi(S_{a^*})\le P_{a^*}$, LP($P_{a^*}$) is feasible and $\TLP\le P_{a^*}$. For $a<a^*$, the predicate fails, i.e., $\Phi(S_a)>P_a$. If $a\le a^*-2$, then also $\Phi(S_a)>P_{a+1}$, since otherwise $\Phi(S_{a+1})\le\Phi(S_a)\le P_{a+1}$ and the predicate would hold for $a+1<a^*$; hence no $T<P_{a^*-1}$ is feasible. It remains to consider the interval $[P_{a^*-1},P_{a^*})$, which exists only if $a^*>1$; on it, LP($T$) is feasible if and only if $T\ge\Phi(S_{a^*-1})$, where $\Phi(S_{a^*-1})>P_{a^*-1}$. Hence, if $a^*>1$ and $\Phi(S_{a^*-1})<P_{a^*}$, the smallest feasible $T$ lies in this interval and $\TLP=\Phi(S_{a^*-1})$, which is returned in line~\ref{Alg1:9}. Otherwise, no $T<P_{a^*}$ is feasible and $\TLP=P_{a^*}$, which is returned in line~\ref{Alg1:8}.

In both cases, the returned optimal solution has $T$-value at most $\TLP$ and uses only variables allowed for $T=\TLP$, so it is feasible for LP($\TLP$). The search uses $O(\log s)=O(\log(nm))$ LP solutions, each of polynomial size.
\end{proof}

Two observations shrink the initial search range $[l,u]$ of line~\ref{Alg1:5} without changing the output; they are not needed for correctness. First, $\TLP\ge\max\{\Phi(S_s),\max_{(i,j,k)}\min_q p_{ijkq}\}$: the bound $\TLP\ge\Phi(S_s)$ holds because every LP($T$) allows only variables in $S_s$, and the bound $\TLP\ge\min_q p_{ijkq}$ holds for every flow $(i,j,k)$ because, by (\ref{interval:a}) and (\ref{interval:e}), LP($T$) needs at least one core $q$ with $p_{ijkq}\le T$ for that flow. Hence the predicate fails for every breakpoint below this value, and the search may start with $l=a_0$, the first breakpoint that is at least this value. Second, the value $\Phi(S_{a_0})$ gives an upper end. If $\Phi(S_{a_0})\le P_{a_0}$, the predicate already holds at $a_0$, so $a^*=a_0$ and no search is needed. Otherwise, $\Phi(S_{a_0})$ is an upper bound on $\TLP$, since LP($T$) with $T=\Phi(S_{a_0})>P_{a_0}$ allows all variables of $S_{a_0}$ and is therefore feasible. This bound need not be tight: a larger set $S_a$ with $P_{a_0}<P_a<\Phi(S_{a_0})$ may have a smaller value $\Phi(S_a)$, so the search is still needed. However, for the first breakpoint $c$ with $P_c\ge\Phi(S_{a_0})$, we have $\Phi(S_c)\le\Phi(S_{a_0})\le P_c$, i.e., the predicate holds at $c$; hence the search may use $u=c$ instead of $u=s$. Moreover, flows of the same port pair that are allowed on the same set of cores are interchangeable in LP($T$): averaging their fractional assignments, weighted by their sizes, preserves (\ref{interval:a})--(\ref{interval:e}). Hence $\Phi(S)$ can be computed from an aggregated LP with one variable per such class and core, i.e., $O(N^2m^2)$ variables independently of the number of flows (see Appendix~\ref{supp:aggregate}).

\subsection{Rounding}
For a feasible solution $(\textbf{x},\textbf{y})$ of LP($\TLP$), we call core $q$ \emph{heavy} for flow $(i,j,k)$ if $x_{ijkq} \ge 1/m$. Every flow has a heavy core, since its $m$ variables sum to one. Algorithm~\ref{Alg5} assigns every flow to one of its heavy cores. With the rule \emph{Max}, it chooses the core with the largest variable (LP-Max). With the rule \emph{Greedy}, it chooses, among the heavy cores, the one that minimizes the current busy time of the two endpoints plus the time added by the flow, which includes $\delta_q$ only if the circuit is new on that core (LP-Greedy). Both rules are analyzed together. A matching-based rounding adapted from Shmoys and Tardos~\cite{shmoys1993approximation} (LP-Match), which assigns every flow to a core in the support of $\textbf{x}$, is described and analyzed in Appendix~\ref{supp:match}.

\begin{algorithm}
\caption{Heavy-Core Rounding (LP-Max and LP-Greedy)}
    \begin{algorithmic}[1]
				\REQUIRE a feasible solution $(\textbf{x},\textbf{y})$ of LP($\TLP$) from Algorithm~\ref{Alg1}; a rule $\in\{$Max, Greedy$\}$
				\ENSURE an integral assignment $\hat{\textbf{x}}$
				\STATE Set $\hat{x}_{ijkq}=0$, $I_{iq}=0$, $O_{jq}=0$ for all $(i, j, k)$, $i$, $j$, $q$; $\mathcal{U}\leftarrow\emptyset$.
				\FOR{each flow $(i, j, k)\in \mathcal{F}$ in non-increasing order of $d_{ijk}$}
						\STATE $H_{ijk} = \{q \in \mathcal{M} : x_{ijkq} \ge 1/m\}$ \COMMENT{heavy cores}
						\STATE $\Delta_{q} = d_{ijk}/r_q + \delta_q\cdot[((i,j),q)\notin\mathcal{U}]$ for $q\in\mathcal{M}$
						\IF{rule = Max}
							\STATE $q^* = \arg \max_{q\in \mathcal{M}} x_{ijkq}$ (ties: smaller $p_{ijkq}$)
						\ELSE
							\STATE $q^* = \arg \min_{q\in H_{ijk}} \left\{\max\{I_{iq}, O_{jq}\}+\Delta_q\right\}$
						\ENDIF
						\STATE $\hat{x}_{ijkq^*}=1$; $I_{iq^*} \mathrel{+}= \Delta_{q^*}$; $O_{jq^*} \mathrel{+}= \Delta_{q^*}$; $\mathcal{U}\leftarrow\mathcal{U}\cup\{((i,j),q^*)\}$
				\ENDFOR
				\STATE \textbf{return} $\hat{\textbf{x}}$.
   \end{algorithmic}
\label{Alg5}
\end{algorithm}

\subsection{Scheduling on Each Core}\label{sec:sched}
Algorithm~\ref{Alg3} schedules the flows assigned to each core. Because of circuit reuse, the flows of the same port pair assigned to the same core are aggregated into one entry of the demand matrix $D^{q}$ of the core, and the single-core scheduler is applied to $D^{q}$; in particular, Reco-Sin and Sunflow serve the aggregated entry of a port pair through circuits of that pair.

\begin{algorithm}
\caption{Scheduling Procedure}
    \begin{algorithmic}[1]
				\STATE Run Algorithm~\ref{Alg5} to obtain $\hat{\textbf{x}}$.
				\FOR{each $q\in \mathcal{M}$ in parallel}
						\STATE $D^{q}_{ij} = \sum_{k: \hat{x}_{ijkq}=1} d_{ijk}$ for all $i,j$.
						\IF{core $q$ is an EPS}
							\STATE Apply BvN decomposition~\cite{marcus1959diagonals, Qiu2015} to $D^{q}$.
						\ELSIF{core $q$ is a not-all-stop OCS}
							\STATE Apply Sunflow~\cite{huang2016} to $D^{q}$.
						\ELSE
							\STATE Apply Reco-Sin~\cite{Tan2021} to $D^{q}$.
						\ENDIF
				\ENDFOR
   \end{algorithmic}
\label{Alg3}
\end{algorithm}

For core $q$ and the assignment $\hat{\textbf{x}}$, define the post-rounding busy time of input port $i$ as
\begin{equation*}
T^{r}_{iq} = \sum_{(i,j,k)\in\Fin{i}} \frac{d_{ijk}}{r_q}\hat{x}_{ijkq} + \delta_q \deg_q(i),
\end{equation*}
where $\deg_q(i)=|\{j : D^{q}_{ij}>0\}|$, and define $T^{r}_{jq}$ for output port $j$ symmetrically. Let $\rho_q=\max\{\max_i \sum_j D^{q}_{ij}, \max_j\sum_i D^{q}_{ij}\}/r_q$ and $\tau_q=\max\{\max_i \deg_q(i), \max_j \deg_q(j)\}\le\tau$. The single-core schedulers satisfy the following bounds (see Appendix~\ref{supp:sched} for our implementations).
\begin{lem}\label{lem:single}
The schedule of core $q$ ends by
(a) $\max_{i,j}\{T^{r}_{iq},T^{r}_{jq}\}$ if $q$ is an EPS core (Theorem~\ref{BvN2});
(b) $\max_{(i,j): D^q_{ij}>0}\{T^{r}_{iq}+T^{r}_{jq}\}$ if $q$ is a not-all-stop OCS core~\cite{huang2016};
(c) $2(\rho_q+\tau_q\delta_q)$ if $q$ is an all-stop OCS core~\cite{Tan2021}.
\end{lem}

\subsection{Analysis}
\begin{lem}\label{thm:5}
Algorithm~\ref{Alg5} (with either rule) guarantees that $T_{iq}^{r}, T_{jq}^{r} \leq \min\left\{\tau, m\right\} \TLP$ for all $i \in \mathcal{I}$, $j \in \mathcal{J}$, and $q\in\mathcal{M}$.
\end{lem}
\begin{proof}
We present the proof for an input port $i$ of a core $q$; output ports are symmetric. Every flow $(i,j,k)$ assigned to $q$ satisfies $x_{ijkq}\ge 1/m$ (with rule Max, because the largest of $m$ values summing to one is at least $1/m$). Hence
\begin{equation}\label{thm5:eq1}
\frac{d_{ijk}}{r_q} \le m\, x_{ijkq}\frac{d_{ijk}}{r_q}.
\end{equation}
Moreover, for every $j$ counted in $\deg_q(i)$, some flow $(i,j,k)$ is assigned to $q$, and by (\ref{interval:y}), $\delta_q \le m\,\delta_q x_{ijkq} \le m\,\delta_q y_{ijq}$. Summing over the flows of $\Fin{i}$ assigned to $q$ and over the partners $j$, and using (\ref{interval:b}), we obtain
\begin{equation*}
T^{r}_{iq} \le m\Big(\sum_{(i,j,k)\in\Fin{i}}\frac{d_{ijk}}{r_q}x_{ijkq}+\delta_q\sum_{j\in\mathcal{J}}y_{ijq}\Big)\le m\,\TLP.
\end{equation*}
Furthermore, every flow assigned to $q$ has $x_{ijkq}>0$ and thus $p_{ijkq} \le \TLP$ by (\ref{interval:e}). Since each partner is counted for at least one such flow, $T^{r}_{iq} \le \sum_{(i,j,k)\in\Fin{i}:\hat{x}_{ijkq}=1} p_{ijkq}\le |\Fin{i}|\,\TLP\le\tau\,\TLP$.
\end{proof}

Let $T^{s}$ denote the makespan of the schedule produced by Algorithm~\ref{Alg3}. We first derive the bounds for networks whose cores are of one type and then for hybrid networks.

\begin{thm}\label{thm:6}
Algorithm~\ref{Alg3} (with either rounding rule) achieves the following approximation guarantees:
\begin{itemize}
\item If all cores are EPS, $T^{s} \le \min\left\{\tau, m\right\}\OPT$.
\item If all cores are not-all-stop OCS, $T^{s} \le 2\min\left\{\tau, m\right\}\OPT$.
\item If all cores are all-stop OCS, $T^{s} \le 2\min\left\{2\tau-1, m+\tau-1, 2m\right\}\OPT$.
\end{itemize}
\end{thm}
\begin{proof}
By Lemma~\ref{lem:relax}, it suffices to bound $T^s$ by the stated factor times $\TLP$. The EPS and not-all-stop cases follow from Lemmas~\ref{lem:single}(a), \ref{lem:single}(b), and~\ref{thm:5}. Consider an all-stop core $q$ that carries at least one flow. Let port $a$ attain $\rho_q$ and port $b$ attain $\tau_q$. Port $a$ carries data and thus has $\deg_q(a)\ge 1$, so $\rho_q \le T^{r}_{aq}-\delta_q$, and $\tau_q\delta_q\le T^{r}_{bq}$. Moreover, $\delta_q<p_{ijkq}\le\TLP$ for any flow assigned to $q$, and $\tau_q\le\tau$. Hence
\begin{equation*}
\rho_q+\tau_q\delta_q \le \min\left\{T^{r}_{aq}+(\tau-1)\TLP,\; T^{r}_{aq}+T^{r}_{bq}\right\}.
\end{equation*}
By Lemma~\ref{thm:5}, the right-hand side is at most $\min\{\min\{\tau,m\}+\tau-1,\,2\min\{\tau,m\}\}\TLP = \min\{2\tau-1, m+\tau-1, 2m\}\TLP$, where we used $2\tau-1\le 2\tau$. Lemma~\ref{lem:single}(c) completes the proof.
\end{proof}

The makespan of a hybrid network is the maximum of the makespans of its cores, and the bound of each core in the proof of Theorem~\ref{thm:6} depends only on the type of that core. Hence the guarantee of a hybrid network is the largest guarantee among the core types it contains. For instance, a network of EPS and not-all-stop OCS cores has the guarantee $2\min\{\tau, m\}$, which equals $4$ for $m=2$ and $\tau\ge 2$, and a network that contains at least one all-stop OCS core has the guarantee $2\min\{2\tau-1, m+\tau-1, 2m\}\le 4m$, which equals $6$ for $m=2$ and $\tau=2$ and $8$ for $m=2$ and $\tau\ge 3$.

\begin{prop}\label{prop:match}
LP-Match assigns every flow to a core $q$ with $x_{ijkq}>0$. Consequently, Theorem~\ref{thm:6} holds for LP-Match with $\min\{\tau,m\}$ replaced by $\tau$, i.e., with the guarantees $\tau$, $2\tau$, and $2(2\tau-1)$.
\end{prop}
The proof is given in Appendix~\ref{supp:match}, which also shows that the matching saturating all flows exists and explains why it controls only input ports.

\subsection{Guarantees, Heuristics, and Certificates}\label{sec:robust}
Simple assignment rules that balance data volumes are natural heuristics for this problem, but in OCS-based networks they have no bounded approximation ratio, because they ignore the cost of circuits.

\begin{prop}\label{prop:unbounded}
Consider the greedy rule that assigns flows in non-increasing order of size, each to the core minimizing $\max\{I_{iq},O_{jq}\}+d_{ijk}/r_q$, where $I_{iq}$ and $O_{jq}$ are the transmission times already assigned to ports $i$ and $j$ on core $q$, and the traffic assignment of Weaver~\cite{Huang2020}. Combined with any single-core schedulers, neither has a bounded approximation ratio for the makespan, even for one coflow, one EPS core, and one not-all-stop OCS core.
\end{prop}
\begin{proof}
Let the EPS core have rate $1$, and let the OCS core have rate $r>n$ and reconfiguration delay $\delta$. One coflow consists of $n$ flows of size $1$ from input port $1$ to $n$ distinct output ports. When the $k$-th flow is considered and the previous ones are on the OCS core, the greedy rule compares $1$ (EPS) with $k/r<1$ (OCS). For Weaver, the flow is critical, because it would increase the bottleneck of the EPS core from $0$ to $1$; Weaver therefore compares the new bottlenecks divided by the rates, which are $1$ (EPS) and $k/r$ (OCS). Both rules place the flow on the OCS core. Hence all $n$ flows are assigned to the OCS core, where port $1$ needs $n$ circuits, and each establishment blocks port $1$ for $\delta$; the makespan is at least $n\delta$. Assigning all flows to the EPS core gives makespan $n$. The ratio is at least $\delta$, which is unbounded.
\end{proof}

A greedy rule that charges $\delta_q$ for every new circuit avoids this example. LP-Greedy uses this rule among the heavy cores of each flow; applying it to all cores, without the LP, gives the heuristic Greedy-R of Section~\ref{sec:Results}. We are not aware of any worst-case guarantee for Greedy-R, and Section~\ref{sec:worst} finds instances on which it is worse than LP-Greedy. Our guarantees, in contrast, hold for every instance, and they can be combined with any heuristic at no loss.

\begin{coro}\label{coro:combine}
Let $\mathcal{S}$ be the schedule of Algorithm~\ref{Alg3} with LP-Max or LP-Greedy and let $\mathcal{S}'$ be any feasible non-splitting schedule, e.g., one produced by a heuristic. The schedule with the smaller makespan satisfies the guarantees of Theorem~\ref{thm:6}. In particular, \emph{LP-Greedy+}, which returns the better of the schedules of LP-Greedy and Greedy-R, has the guarantees of Table~\ref{tab:guarantees}.
\end{coro}
\begin{proof}
Its makespan is at most that of $\mathcal{S}$, which satisfies Theorem~\ref{thm:6}.
\end{proof}

\textbf{Certificates.} Since $\TLP\le\OPT$ (Lemma~\ref{lem:relax}) and Algorithm~\ref{Alg1} computes $\TLP$ exactly, our algorithms also return, for every instance, the a-posteriori bound $T^{s}/\TLP$ on the ratio between the makespan of the returned schedule and the optimum. This instance-specific bound is usually much smaller than the worst-case guarantee (Section~\ref{sec:Results}); a heuristic alone gives no such statement.

\section{Approximation Algorithm for Minimizing the Total Weighted Completion Time}\label{sec:Algorithm5}
This section extends the framework to the total weighted coflow completion time with release times. Inspired by~\cite{CHUDAK1999323}, the time horizon is divided into geometrically increasing intervals; an LP decides, fractionally, in which interval and on which core each flow completes; flows are then grouped by interval, each group is scheduled by the makespan framework, and the groups are executed one after another.

\subsection{Time Normalization and Intervals}
Without loss of generality, we scale time so that $\min_{(i,j,k)\in\mathcal{F},q\in\mathcal{M}} d_{ijk}/r_q = 1$; every flow then completes at time at least $1$ in every schedule. Let $t_{0}=1$ and $t_{l}=2^{l}$, and define the intervals $I_1=[t_0,t_1]$ and $I_l=(t_{l-1},t_l]$ for $l=2,\ldots,L$. To choose $L$, let $H=\max_k a_k+\sum_{(i,j,k)\in\mathcal{F}}\min_{q} p_{ijkq}$: transmitting the flows one at a time after $\max_k a_k$, each on its fastest core, is a feasible schedule in which all coflows complete by $H$, so $\OPT\le W H$ with $W=\sum_k w_k$. In an optimal schedule, $w_k C_k\le\OPT$, so every coflow completes by $WH/\min_k w_k$. We set
\begin{equation*}
L=\left\lceil \log_2 \left(W H/\min_{k} w_k\right)\right\rceil,
\end{equation*}
which is polynomial in the input size. (Any feasible schedule may replace $WH$ as an upper bound on $\OPT$.)

\subsection{LP Relaxation}
For flow $(i,j,k)$ and core $q$, let $e_{ijkq}=\max\{a_k,\delta_q\}+d_{ijk}/r_q$ be the earliest time at which the flow can complete on core $q$: its data cannot be sent before $a_k$, and on an OCS core its circuit must be established first. The variable $x_{ijkql}$ indicates whether flow $(i,j,k)$ is sent through core $q$ and completes in interval $I_l$, and $y_{ijql}$ indicates whether circuit $(i,j)$ is used on OCS core $q$ by time $t_l$.
\begin{subequations}\label{coflow:main:wc}
\begin{align}
&\min \ \sum_{k\in \mathcal{K}}w_k C_{k} \quad \text{subject to} \tag{\ref{coflow:main:wc}}\\ 
& \sum_{q\in \mathcal{M}} \sum_{l=1}^{L} x_{ijkql}=1, \qquad\qquad\qquad\qquad \forall (i, j, k), \label{coflow:wc:a}\\ 
& \sum_{q\in \mathcal{M}} e_{ijkq}\sum_{l=1}^{L}x_{ijkql}\leq C_{ijk}, \qquad\qquad\quad\ \forall (i, j, k), \label{coflow:wc:b}\\
& \sum_{u=1}^{l} \sum_{(i,j,k)\in \Fin{i}} \tfrac{d_{ijk}}{r_q} x_{ijkqu} +\delta_q\sum_{j\in\mathcal{J}}y_{ijql} \leq t_{l}, \ \ \forall i, q, l, \label{coflow:wc:f}\\ 
& \sum_{u=1}^{l} \sum_{(i,j,k)\in \Fout{j}} \tfrac{d_{ijk}}{r_q} x_{ijkqu} +\delta_q\sum_{i\in\mathcal{I}}y_{ijql} \leq t_{l}, \ \ \forall j, q, l, \label{coflow:wc:g}\\ 
& \sum_{u=1}^{l} x_{ijkqu} \leq y_{ijql}, \qquad\qquad \forall (i, j, k), q\in\mathcal{M}_o, l, \label{coflow:wc:y}\\ 
& \sum_{l=1}^{L}t_{l-1} \sum_{q\in \mathcal{M}} x_{ijkql} \leq C_{ijk}, \qquad\qquad\quad\ \forall (i, j, k), \label{coflow:wc:h}\\ 
& C_{ijk} \leq C_{k}, \qquad\qquad\qquad\qquad\qquad\qquad\quad\ \forall (i, j, k), \label{coflow:wc:j}\\ 
& x_{ijkql}, y_{ijql}, C_{ijk}, C_{k}\geq 0, \qquad\qquad\qquad \forall (i,j,k), q, l, \label{coflow:wc:i} \\
& x_{ijkql} = 0 \ \text{ if } e_{ijkq}>t_{l}, \qquad\qquad\qquad\ \ \forall (i, j, k), q, l. \label{coflow:wc:k} 
\end{align}   
\end{subequations}
Constraint (\ref{coflow:wc:a}) assigns every flow to a core and an interval. Constraint (\ref{coflow:wc:b}) bounds the completion time of a flow by its earliest completion time on its core. Constraints (\ref{coflow:wc:f}) and (\ref{coflow:wc:g}) bound the busy time of every port of every core up to $t_l$, and (\ref{coflow:wc:y}) charges one reconfiguration delay per circuit, as in LP($T$). Constraint (\ref{coflow:wc:h}) states that a flow completing in $I_l$ completes after $t_{l-1}$, (\ref{coflow:wc:j}) defines the coflow completion time, and (\ref{coflow:wc:k}) excludes intervals that end before the flow can complete.

\begin{lem}\label{lem:wrelax}
The optimal value of LP~(\ref{coflow:main:wc}) is at most $\OPT$.
\end{lem}
\begin{proof}
Take an optimal schedule with completion times $C^{*}_{ijk}$ and $C^{*}_k$. By the choice of $L$ and the normalization, $C^*_{ijk}\in[t_0,t_L]$, so it lies in exactly one interval $I_{l}$. Set $x_{ijkql}=1$ if flow $(i,j,k)$ is sent through core $q$ and $C^*_{ijk}\in I_l$, and $0$ otherwise; set $y_{ijql}=1$ if some flow of pair $(i,j)$ sent through OCS core $q$ completes by $t_l$; and set $C_{ijk}=C^*_{ijk}$ and $C_k=C^*_k$. Then (\ref{coflow:wc:a}), (\ref{coflow:wc:y}), (\ref{coflow:wc:j}), and (\ref{coflow:wc:i}) hold by construction. Constraint (\ref{coflow:wc:b}) holds because a flow on core $q$ cannot complete before $e_{ijkq}$, and (\ref{coflow:wc:k}) holds because $C^*_{ijk}\le t_l$ implies $e_{ijkq}\le t_l$. Constraint (\ref{coflow:wc:h}) holds because $C^*_{ijk}> t_{l-1}$ for $l\ge 2$ and $C^*_{ijk}\ge t_0$ for $l=1$. For (\ref{coflow:wc:f}), consider port $i$ of core $q$ and the flows of $\Fin{i}$ on $q$ that complete by $t_l$. Exactly as in the proof of Lemma~\ref{lem:relax}, the transmission of their data and one establishment of each of their circuits are disjoint busy periods of port $i$ within $[0,t_l]$, so their total length, which is the left-hand side of (\ref{coflow:wc:f}), is at most $t_l$. Output ports are symmetric. The objective value of this solution is $\sum_k w_k C^*_k=\OPT$.
\end{proof}

\subsection{Rounding and Grouping}
Let $(\textbf{x},\textbf{y},\textbf{C})$ be an optimal solution of LP~(\ref{coflow:main:wc}). Algorithm~\ref{Alg7} sets, for each flow $(i, j, k)$, $g(i, j, k) = \max \left\{g_{1}, g_{2}\right\}$, where
\begin{align*}
g_{1} &= \min \Big\{ g \ \Big| \ \sum_{l=1}^{g}\sum_{q\in \mathcal{M}} x_{ijkql} \geq \frac{1}{2} \Big\},\\
g_{2} &= \min \left\{ g \ \middle| \ C_{ijk} \leq t_{g} \right\},
\end{align*}
and assigns the flow to the core $q$ with the largest accumulated value $X_{ijkq} = \sum_{l=1}^{g(i, j, k)} x_{ijkql}$. The flows assigned to core $q$ in group $g$ form the set
\begin{equation*}
\mathcal{A}_{q, g} = \left\{ (i, j, k) \ \middle| \ g(i, j, k) = g \text{ and } \hat{x}_{ijkq} = 1 \right\}.
\end{equation*}

\begin{algorithm}
\caption{Rounding Procedure for the Weighted Completion Time}
    \begin{algorithmic}[1]
				\STATE Solve LP (\ref{coflow:main:wc}) to obtain $\textbf{C}$ and $\textbf{x}$; set $\mathcal{A}_{q,g}=\emptyset$ for all $q$, $g$.
				\FOR{each flow $(i, j, k)\in \mathcal{F}$} \label{Alg4_line1}
				    \STATE $g_{1} = \min \left\{ g \ \middle| \ \sum_{l=1}^{g}\sum_{q\in \mathcal{M}} x_{ijkql} \geq \frac{1}{2} \right\}$.
						\STATE $g_{2} = \min \left\{ g \ \middle| \ C_{ijk} \leq t_{g} \right\}$.
						\STATE $g(i, j, k)=\max \left\{g_{1}, g_{2}\right\}$.
						\STATE $X_{ijkq}= \sum_{l=1}^{g(i, j, k)} x_{ijkql}$ for each $q\in \mathcal{M}$.
						\STATE $q^* = \arg \max_{q\in \mathcal{M}} X_{ijkq}$
						\STATE $\mathcal{A}_{q^*, g(i, j, k)}=\mathcal{A}_{q^*, g(i, j, k)} \cup \left\{(i, j, k)\right\}$; $\hat{x}_{ijkq^*}=1$
				\ENDFOR \label{Alg4_line2}
				\STATE \textbf{return} $\{\mathcal{A}_{q,g}\}$.
   \end{algorithmic}
\label{Alg7}
\end{algorithm}

\subsection{Scheduling with Release Times}
Algorithm~\ref{Alg8} executes the groups on every core in increasing order of $g$. Group $g$ starts on core $q$ when the core has finished group $g-1$ and all coflows of the group have been released:
\begin{equation*}
S_{q,g}=\max\Big\{E_{q,g-1},\ \max_{(i,j,k)\in\mathcal{A}_{q,g}} a_k\Big\},\qquad E_{q,g}=S_{q,g}+C^{g}_{q},
\end{equation*}
where $E_{q,0}=0$ and $C^{g}_{q}$ is the length of the single-core schedule of group $g$ on core $q$ (Algorithm~\ref{Alg3} applied to the aggregated demand of $\mathcal{A}_{q,g}$). Inside a group, the flows that share a circuit are sent in non-decreasing order of $C_k/w_k$ of their coflows, and the single-core schedulers serve port pairs and configurations in the same order; this ordering does not change $C^{g}_{q}$ and therefore does not affect the analysis.

\begin{algorithm}
\caption{Scheduling Procedure for the Weighted Completion Time}
    \begin{algorithmic}[1]
				\STATE Run Algorithm~\ref{Alg7} to obtain $\{\mathcal{A}_{q,g}\}$.
				\FOR{each $q\in \mathcal{M}$ in parallel}
				    \STATE $E_{q,0}\leftarrow 0$
					\FOR{$g=1, 2, \ldots, L$ with $\mathcal{A}_{q,g}\neq\emptyset$}
					    \STATE $S_{q,g}\leftarrow\max\{E_{q,g-1}, \max_{(i,j,k)\in\mathcal{A}_{q,g}} a_k\}$
						\STATE Starting at $S_{q,g}$, schedule $\mathcal{A}_{q,g}$ on core $q$ by BvN, Sunflow, or Reco-Sin according to the type of $q$ (as in Algorithm~\ref{Alg3}).
						\STATE $E_{q,g}\leftarrow S_{q,g}+C^{g}_{q}$
					\ENDFOR
				\ENDFOR
   \end{algorithmic}
\label{Alg8}
\end{algorithm}

\subsection{Analysis}
For core $q$ and group $g$, define the busy time of input port $i$ in group $g$ as
\begin{equation*}
T^{g}_{iq}=\sum_{(i,j,k)\in\Fin{i}\cap\mathcal{A}_{q,g}}\frac{d_{ijk}}{r_q}+\delta_q\deg_{q,g}(i),
\end{equation*}
where $\deg_{q,g}(i)$ is the number of distinct partners of port $i$ among the flows of $\mathcal{A}_{q,g}$, and define $T^{g}_{jq}$ for output port $j$ symmetrically. These are the post-rounding busy times of Section~\ref{sec:sched} restricted to group $g$.

\begin{lem}\label{lem:3}
Algorithm~\ref{Alg7} guarantees that $T_{iq}^{g}, T_{jq}^{g} \leq \min\{\tau, 2m\} t_{g}$ for all $i \in \mathcal{I}$, $j \in \mathcal{J}$, $q\in\mathcal{M}$, and $g\in [1, L]$.
\end{lem}
\begin{proof}
We present the proof for an input port $i$; output ports are symmetric. Let $(i,j,k)\in\mathcal{A}_{q,g}$. Since $g\ge g_1$, we have $\sum_{q'}X_{ijkq'}\ge 1/2$, and since $q$ maximizes $X_{ijkq'}$ over the $m$ cores, $X_{ijkq}\ge 1/(2m)$. Hence $d_{ijk}/r_q \le 2m\,X_{ijkq}\,d_{ijk}/r_q = 2m\sum_{u=1}^{g}x_{ijkqu}\,d_{ijk}/r_q$, and, for every partner $j$ of port $i$ in $\mathcal{A}_{q,g}$, $\delta_q\le 2m\,\delta_q X_{ijkq} \le 2m\,\delta_q y_{ijqg}$ by (\ref{coflow:wc:y}) with $l=g$. Summing and using (\ref{coflow:wc:f}) with $l=g$ gives $T^{g}_{iq}\le 2m\,t_g$. Moreover, $X_{ijkq}>0$ implies $x_{ijkqu}>0$ for some $u\le g$, so $p_{ijkq}\le e_{ijkq}\le t_u\le t_g$ by (\ref{coflow:wc:k}); hence $T^{g}_{iq}\le|\Fin{i}|\,t_g\le\tau\, t_g$.
\end{proof}

\begin{lem}\label{lem:group}
For every core $q$ and group $g$, $C^{g}_{q}\le R_q\, t_g$, where $R_q=R_{E}=\min\{\tau,2m\}$ if $q$ is an EPS core, $R_q=R_{N}=2\min\{\tau,2m\}$ if $q$ is a not-all-stop OCS core, and $R_q=R_{A}=2\min\{2\tau-1, 2m+\tau-1, 4m\}$ if $q$ is an all-stop OCS core.
\end{lem}
\begin{proof}
The proof of Theorem~\ref{thm:6} applies verbatim to the flows of $\mathcal{A}_{q,g}$, with $T^{g}_{iq}$ and $T^{g}_{jq}$ in place of $T^{r}_{iq}$ and $T^{r}_{jq}$, Lemma~\ref{lem:3} in place of Lemma~\ref{thm:5}, and $t_g$ in place of $\TLP$; note that $\delta_q< p_{ijkq}\le t_g$ for every flow of $\mathcal{A}_{q,g}$.
\end{proof}

\begin{lem}\label{lem:tg}
For every flow $(i,j,k)$ with $g=g(i,j,k)$, $t_g\le 4C_{ijk}$.
\end{lem}
\begin{proof}
We first show $C_{ijk}\ge 1$. By the normalization, $e_{ijkq}\ge d_{ijk}/r_q\ge 1$ for every core $q$. By (\ref{coflow:wc:a}), the coefficients $\sum_{l=1}^{L}x_{ijkql}$ over $q\in\mathcal{M}$ are nonnegative and sum to $1$, so the left-hand side of (\ref{coflow:wc:b}) is a convex combination of the values $e_{ijkq}$ and is at least $1$; hence $C_{ijk}\ge 1$. If $g=1$, then $t_g=t_1=2\le 2C_{ijk}\le 4C_{ijk}$. Now let $g\ge 2$. If $g=g_2$, then $t_{g-1}<C_{ijk}$ by the minimality of $g_2$, so $t_g=2t_{g-1}<2C_{ijk}$. If $g=g_1$, then $\sum_{l=1}^{g-1}\sum_q x_{ijkql}<1/2$ by the minimality of $g_1$, and by (\ref{coflow:wc:a}) and (\ref{coflow:wc:h}),
\begin{equation*}
\frac{t_{g-1}}{2} \le t_{g-1}\sum_{l=g}^{L}\sum_{q\in \mathcal{M}} x_{ijkql} \le \sum_{l=1}^{L}\sum_{q\in \mathcal{M}} t_{l-1} x_{ijkql} \le C_{ijk},
\end{equation*}
so $t_g=2t_{g-1}\le 4C_{ijk}$.
\end{proof}

\begin{lem}\label{lem:start}
For every core $q$ and group $g$, $E_{q,g}\le t_g+R_q\sum_{l=1}^{g}t_l < (2R_q+1)\,t_g$. If all release times are zero, $E_{q,g}\le R_q\sum_{l=1}^{g}t_l < 2R_q\, t_g$.
\end{lem}
\begin{proof}
Every flow of $\mathcal{A}_{q,g}$ satisfies $a_k\le e_{ijkq}\le t_g$ by (\ref{coflow:wc:k}), as shown in the proof of Lemma~\ref{lem:3}. Hence $S_{q,g}\le\max\{E_{q,g-1},t_g\}$. By induction on $g$ (with $E_{q,0}=0$) and Lemma~\ref{lem:group},
\begin{equation*}
E_{q,g}\le\max\Big\{t_{g-1}+R_q\sum_{l=1}^{g-1}t_l,\ t_g\Big\}+R_q t_g\le t_g+R_q\sum_{l=1}^{g}t_l.
\end{equation*}
With zero release times, $S_{q,g}=E_{q,g-1}$ and the same induction gives $E_{q,g}\le R_q\sum_{l=1}^{g}t_l$. Finally, $\sum_{l=1}^{g}t_l=2^{g+1}-2<2t_g$.
\end{proof}

Let $\hat{C}_{k}$ and $\hat{C}_{ijk}$ denote the completion times in the schedule of Algorithm~\ref{Alg8}, and let $R=\max_{q}R_q$ over the core types present in the network.

\begin{thm}\label{thm:9}
Algorithm~\ref{Alg8} achieves $\sum_k w_k\hat{C}_{k}\le (8R+4)\,\OPT$ for arbitrary release times and $\sum_k w_k\hat{C}_{k}\le 8R\,\OPT$ for zero release times. In particular, for zero release times the guarantee is $8\min\{\tau, 2 m\}$ if all cores are EPS, $16\min\{\tau, 2 m\}$ if all cores are not-all-stop OCS, and $16\min\{2\tau-1, 2m+\tau-1, 4m\}$ if all cores are all-stop OCS; arbitrary release times add $4$.
\end{thm}
\begin{proof}
Consider flow $(i,j,k)$ in group $g$ on core $q$. It completes by $E_{q,g}$, so by Lemmas~\ref{lem:start} and~\ref{lem:tg}, $\hat{C}_{ijk}<(2R+1)t_g\le 4(2R+1)C_{ijk}\le(8R+4)C_k$, using (\ref{coflow:wc:j}). Taking the maximum over the flows of coflow $k$, multiplying by $w_k$, summing over $k$, and applying Lemma~\ref{lem:wrelax} proves the first claim. For zero release times, $\hat{C}_{ijk}<2Rt_g\le 8RC_k$ in the same way.
\end{proof}

In a hybrid network, $R$ is determined by the weakest architecture present. Since every group starts only after its coflows are released, the schedule is feasible for arbitrary release times, and the additive term $4$ accounts for the waiting. As in Corollary~\ref{coro:combine}, returning the better of the schedule of Algorithm~\ref{Alg8} and any feasible heuristic schedule (\emph{LP-Group+}) keeps the guarantees of Theorem~\ref{thm:9}, and the optimal value of LP~(\ref{coflow:main:wc}) certifies the quality of the returned schedule on every instance.

\subsection{Discussion of the Guarantees}\label{sec:discussion}
\textbf{What the bounds depend on.} The guarantees do not depend on the number of ports $N$, the number of coflows, the flow sizes, the rates, or the reconfiguration delays. They depend only on the number of cores $m$ and on $\tau$; every makespan bound is at most $4m$ regardless of $\tau$ and at most $2(2\tau-1)$ regardless of $m$, and the weighted bounds are $O(\min\{m,\tau\})$ as well. In practice, $m$ is the number of parallel switching planes, which is a small constant, whereas $\tau$ grows with the traffic; in large networks the $m$-dependent part of the bound is therefore the relevant one. When few flows share a port ($\tau \le m$), the bounds are those of the single-core schedulers up to the factor $\tau$.

\textbf{Limits of the analysis.} The bounds are upper bounds and we do not claim that they are tight. Every analysis that compares with $\TLP$ loses at least a factor close to $2$: for $m$ EPS cores of rate $1$ and one input port with $m+1$ flows of size $m$ to distinct outputs, $x_{ijkq}=1/m$ is feasible for $T=m+1$, whereas any non-splitting schedule needs $2m$; hence $\OPT/\TLP \ge 2m/(m+1)$. The gap between this lower bound and our upper bounds is open.

\textbf{Comparison under comparable assumptions.} Table~\ref{tab:compare} compares our guarantees with previous results in their settings. For $m=1$, LP rounding is trivial and our algorithm coincides with the single-core scheduler, so it inherits the guarantees $1$, $2$, and $2$ of BvN, Sunflow, and Reco-Sin. For two cores, \cite{10411848, 11200997} allow splitting and compare with the splittable optimum, which is a stronger benchmark; our bounds hold without splitting and are independent of the number of coflows and of the weights. For EPS-only networks, the bounds of~\cite{CHEN2023104752} are $O(\log m/\log\log m)$ and thus asymptotically smaller than ours, but their constants are large. The makespan bound $8K+4\gamma$ ($4K+2\gamma$ if the ratio of the largest to the smallest rate is less than $m$) is at least $m$ for every $m\le 43$ ($m\le 17$). The weighted bound $64K+32\gamma$ exceeds $199$ for every $m$ and is at least $16m+4$ for every $m\le 17$; hence it is smaller than our bound $8\min\{\tau,2m\}+4$ only if $\tau\ge 25$ and $m\ge 18$. For the small numbers of cores of practical networks, our bounds are therefore smaller. Our contribution in this setting is that one algorithm and one analysis also cover OCS cores with circuit reuse, all-stop cores, and mixtures of all three types.

\begin{table}[!t]
\caption{Guarantees under comparable settings (WCT: total weighted completion time; ``split'': flows may be split; $\gamma=\log m/\log\log m$ and $K=\lceil\log_\gamma m\rceil$, so that $8K+4\gamma$ and $64K+32\gamma$ are $O(\log m/\log\log m)$).}
\label{tab:compare}
\centering
\renewcommand{\arraystretch}{1.2}
\begin{tabular}{|>{\raggedright\arraybackslash}p{2.0in}|>{\raggedright\arraybackslash}p{2.4in}|>{\raggedright\arraybackslash}p{1.5in}|}
\hline
Setting & Previous & This paper \\ \hline
$m=1$, makespan & $1$ (EPS)~\cite{Qiu2015}, $2$ (n.-a.-s.)~\cite{huang2016}, $2$ (a.-s.)~\cite{Tan2021} & same (coincides) \\ \hline
EPS cores, makespan of one coflow & $\frac{\sum_q r_q}{\max_q r_q}\le m$~\cite{Huang2020}; $8K{+}4\gamma$~\cite{CHEN2023104752} & $\min\{\tau, m\}$ \\ \hline
EPS cores, WCT & $64K{+}32\gamma$~\cite{CHEN2023104752} & $8\min\{\tau,2m\}{+}4$ \\ \hline
EPS + a.-s. OCS ($m{=}2$), one coflow & $2\tau'$, split~\cite{10411848} & $\le 8$, no split \\ \hline
EPS + OCS ($m{=}2$), one coflow & $\lambda+1$, split~\cite{11200997} & $\le 4$ (n.-a.-s.), $\le 8$ (a.-s.) \\ \hline
EPS + a.-s. OCS ($m{=}2$), WCT & $4M\frac{w_{\max}}{w_{\min}}\tau'_{\max}$, split~\cite{10411848} & $\le 16\cdot 8{+}4$ \\ \hline
$m\ge 2$, all-stop or three types & -- & Table~\ref{tab:guarantees} \\ \hline
\end{tabular}
\end{table}

\section{Performance Evaluation}\label{sec:Results}
This section evaluates the proposed algorithms by simulation. The implementation follows the listings and the model of this paper, including circuit reuse and the exact computation of $\TLP$ (Algorithm~\ref{Alg1}).

\subsection{Simulation Setup}
\textbf{Implementation.} All algorithms are implemented in Python~3.13 with NumPy~2.5 and SciPy~1.18; LPs are solved by the HiGHS solver through SciPy, and the MILPs of Sections~\ref{sec:exact} and~\ref{sec:worst} by the HiGHS branch-and-bound solver. All experiments were run on a desktop computer with an Intel Core Ultra~7 270K Plus processor (24 cores) and 32~GB of RAM running Windows~10. The source code and the raw results are available as ancillary files of this arXiv submission.

\textbf{Units and parameters.} Data are measured in megabits (Mb) and time in milliseconds (ms), so rates are in Mb/ms, i.e., Gb/s. The link rate of an EPS core is drawn uniformly from $\{1,2,3\}$~Gb/s per port, which models an oversubscribed packet-switched plane, and that of an OCS core from $\{150,\ldots,300\}$~Gb/s, which corresponds to 100--400~Gb/s optical links. The reconfiguration delay $\delta_q$ of an OCS core is drawn uniformly from $\{1,\ldots,10\}$~ms, a millisecond-scale delay typical of MEMS-based OCSs; Section~\ref{sec:res:delta} scales it by $0.01$, $0.1$, and $10$, i.e., from $10~\mu$s to $100$~ms, to cover faster optical switches as well as slower ones. With these values, a $10$~Mb flow needs about $0.05$~ms on an OCS core, far less than the establishment of its circuit, whereas a $1000$~Mb flow needs about $5$~ms on an OCS core but several hundred ms on an EPS core.

\textbf{Workloads.} Following~\cite{shafiee2018improved}, each coflow is described by $(W_{\min}, W_{\max}, L_{\min}, L_{\max})$: the numbers of active input and output ports are drawn uniformly from $[W_{\min}, W_{\max}]$, a flow is created for every pair of active ports, and every flow size is drawn uniformly from $[L_{\min}, L_{\max}]$~Mb. The four profiles $(1, 5, 1, 10)$, $(1, 5, 10, 1000)$, $(5, N, 1, 10)$, and $(5, N, 10, 1000)$ are mixed with probabilities $41\%$, $29\%$, $9\%$, and $21\%$; other mixtures are evaluated in Appendix~\ref{supp:exp}.

\textbf{Network.} By default, $N=10$, there are $100$ coflows (about 2311 flows, $\tau\approx 269$), and $m=10$ cores, of which $20\%$ are EPS, $40\%$ not-all-stop OCS, and $40\%$ all-stop OCS. Each data point averages 200 independent random instances.

\subsection{Compared Schemes}
All schemes apply the same single-core schedulers (Algorithm~\ref{Alg3}) to every core and use the same accounting of reconfiguration delays (one circuit establishment per port pair and core, with circuit reuse); they differ only in how flows are assigned to cores.
\begin{itemize}
\item \textbf{LP-Max} and \textbf{LP-Greedy} (Algorithm~\ref{Alg5}) and \textbf{LP-Match} (Appendix~\ref{supp:match}): our LP-based schemes, applied to the optimal LP solution returned by Algorithm~\ref{Alg1}.
\item \textbf{LP-Greedy+}: the combination of Section~\ref{sec:robust}, which returns the better of the schedules of LP-Greedy and Greedy-R and keeps the guarantee of LP-Greedy.
\item \textbf{Weaver}~\cite{Huang2020}: the traffic assignment of Weaver as designed for EPS networks. For core $q$, let $I_{iq}$ and $O_{jq}$ be the total transmission time $\sum d_{ijk}/r_q$ of the flows already assigned to input port $i$ and output port $j$, and let $L_q$ be the largest of these times over all ports. Flows are considered in non-increasing order of size. Flow $(i,j,k)$ is \emph{critical} if assigning it to some core $q$ would increase $L_q$; a critical flow is placed on the core that minimizes its new bottleneck $\max\{L_q, I_{iq}+d_{ijk}/r_q, O_{jq}+d_{ijk}/r_q\}$. Any other flow is placed on the core that minimizes the busy ratio $\max\{I_{iq}, O_{jq}\}+d_{ijk}/r_q$. (Weaver states these rules with data volumes divided by $r_q$, which is equivalent.) Reconfiguration delays are ignored when assigning, but they are paid in the schedule.
\item \textbf{Greedy}: a data-only greedy baseline. Flows are considered in non-increasing order of size, and every flow is placed on the core that minimizes $\max\{I_{iq}, O_{jq}\}+d_{ijk}/r_q$, with $I_{iq}$ and $O_{jq}$ as for Weaver, i.e., the busy-ratio rule of Weaver is applied to every flow; reconfiguration delays are ignored when assigning.
\item \textbf{Weaver-R} and \textbf{Greedy-R}: reconfiguration-aware variants of Weaver and Greedy that we designed for this evaluation. As in Algorithm~\ref{Alg5}, $I_{iq}$ and $O_{jq}$ also include $\delta_q$ for every circuit established on core $q$, and $d_{ijk}/r_q$ in the rules above is replaced by $\Delta_q$, which is $d_{ijk}/r_q+\delta_q$ if the port pair of the flow has no circuit on core $q$ yet and $d_{ijk}/r_q$ otherwise. Greedy-R uses the rule that LP-Greedy applies among the heavy cores of a flow, but it chooses among all cores and does not use the LP. No worst-case guarantee is known for these variants.
\item \textbf{LP-Split}: a splitting-based reference scheme. Every flow is split according to an optimal solution of LP($T$) without constraint (\ref{interval:e}), and every part is sent through its core as a separate flow.
\end{itemize}
The hybrid schemes of~\cite{10411848, 11200997} schedule a single coflow on one EPS and one OCS core and split the traffic between them, so they cannot be applied to $m$ cores. Instead, we compare with $\Phi_{\mathrm{split}}=\Phi(S_s)$, the optimal value of LP($T$) without (\ref{interval:e}); by Lemma~\ref{lem:splitlb} in Appendix~\ref{supp:aggregate}, it is a lower bound on the makespan of every splitting-based schedule, including those of~\cite{10411848, 11200997} in their settings.

\textbf{Metric.} We report the \emph{normalized makespan}, i.e., the makespan divided by $\TLP$. Since $\TLP\le\OPT$, it is an upper bound on the ratio between the makespan and the optimum of the instance, and for the LP-based schemes it is a certificate that is available at run time (Section~\ref{sec:robust}). It is not the approximation ratio, which is a worst-case guarantee over all instances; Sections~\ref{sec:exact} and~\ref{sec:worst} compare with exact optima and search for bad instances.

\subsection{Makespan}
\textbf{Number of cores.} Fig.~\ref{fig:ratio3} shows the normalized makespan for $5$ to $25$ cores. All schemes stay far below the proven bounds (e.g., $2\min\{2\tau-1, m+\tau-1, 2m\}=40$ for $m=10$). Schemes that ignore reconfiguration delays degrade quickly with $m$: Greedy grows from 2.46 for $m=5$ to 6.18 for $m=25$, and Weaver to 6.63, because they place flows on fast OCS cores regardless of the circuits they open. LP-Greedy achieves 1.87, 2.14, and 2.82 for $m=5$, $10$, and $25$, i.e., up to 54\% less than Greedy. The reconfiguration-aware heuristics Greedy-R and Weaver-R, which serve as stronger baselines, are close to LP-Greedy: Greedy-R is slightly worse for $m=5$ (1.93) and better for $m\ge 10$ (2.44 for $m=25$). LP-Max and LP-Match degrade as $m$ grows (4.56 and 3.99 for $m=25$): an optimal LP solution spreads the flows of every port pair over many cores, since the LP charges a circuit only fractionally, and assigning every flow to its largest fraction concentrates whole classes of flows on one core. The combination LP-Greedy+ is the best non-splitting scheme for every $m$ (1.86, 2.08, and 2.44) while it keeps the worst-case guarantee of LP-Greedy. Moreover, its normalized makespan is a certificate: for $m=5$ and $m=10$, it proves in 92\% and 26\% of the instances that the schedule is within a factor $2$ of the optimum, and for $m=10$ in 100\% within $2.5$; none of the heuristics can give such a statement without solving the LP.

\begin{figure}[!t]
    \centering
        \includegraphics[width=3.4in]{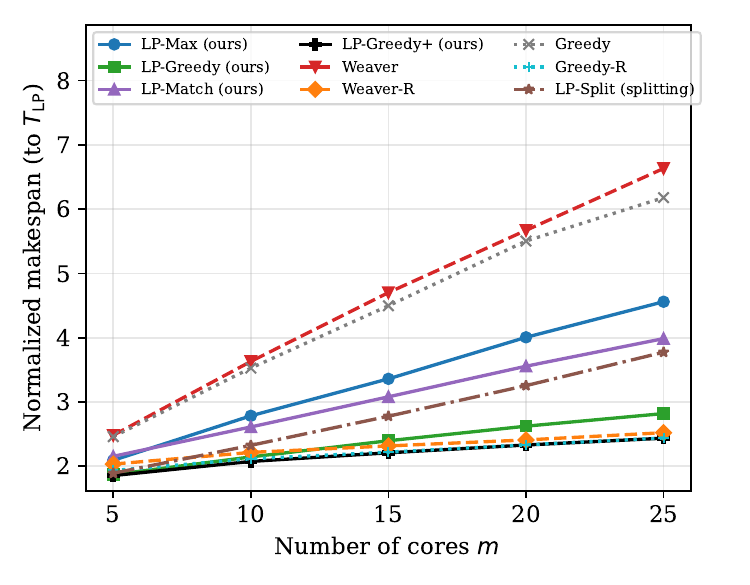}
    \caption{Normalized makespan for $5$ to $25$ cores (mean of $100$ instances per point).}
    \label{fig:ratio3}
\end{figure}

\begin{figure}[!t]
    \centering
        \includegraphics[width=3.4in]{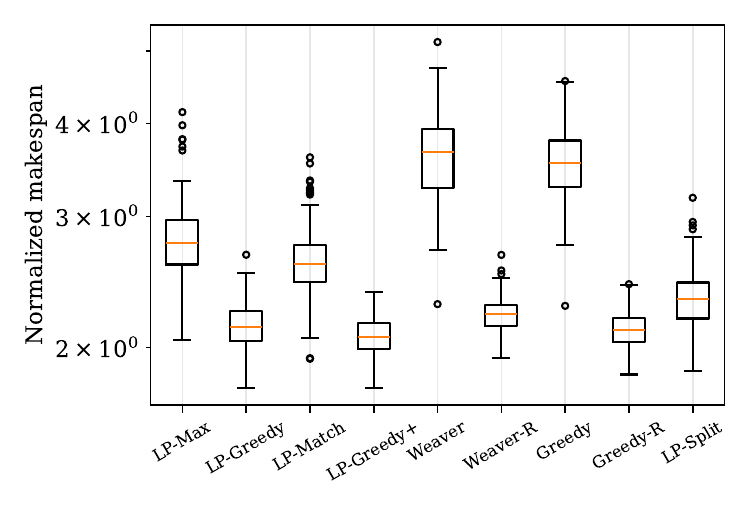}
    \caption{Distribution of the normalized makespan for $m=10$ (logarithmic scale).}
    \label{fig:ratio4}
\end{figure}

\textbf{Distribution.} Fig.~\ref{fig:ratio4} and Table~\ref{table:ratio4} show the distribution for $m=10$, together with the mean makespan in ms; the mean of $\TLP$ is 56.1~ms. The medians of LP-Greedy and Greedy-R are 2.13 and 2.11, and that of LP-Greedy+ is 2.07. The largest normalized makespans over the $200$ instances are 2.37 for LP-Greedy+, 2.66 for LP-Greedy, and 2.43 for Greedy-R, compared with 4.56 for Greedy and 5.14 for Weaver.

\begin{table}[!t]
\caption{Quartiles, maximum, and minimum of the normalized makespan, and mean makespan (ms), for $m=10$.}
\centering
\setlength{\tabcolsep}{4.5pt}
\begin{tabular}{|l|c|c|c|c|c|c|}
\hline
         & $Q_1$ & $Q_2$ & $Q_3$ & max & min & ms \\ \hline
LP-Max & 2.59 & 2.76 & 2.97 & 4.14 & 2.05 & 156.2 \\ \hline
LP-Greedy & 2.04 & 2.13 & 2.24 & 2.66 & 1.77 & 120.2 \\ \hline
LP-Match & 2.45 & 2.59 & 2.75 & 3.60 & 1.93 & 146.0 \\ \hline
LP-Greedy+ & 1.99 & 2.07 & 2.16 & 2.37 & 1.77 & 116.3 \\ \hline
Weaver & 3.28 & 3.66 & 3.92 & 5.14 & 2.29 & 202.2 \\ \hline
Weaver-R & 2.14 & 2.22 & 2.28 & 2.66 & 1.94 & 124.2 \\ \hline
Greedy & 3.28 & 3.54 & 3.79 & 4.56 & 2.28 & 196.5 \\ \hline
Greedy-R & 2.03 & 2.11 & 2.19 & 2.43 & 1.84 & 118.7 \\ \hline
LP-Split & 2.19 & 2.32 & 2.45 & 3.18 & 1.86 & 130.0 \\ \hline
\end{tabular}
\label{table:ratio4}
\end{table}

\textbf{Reconfiguration delay.}\label{sec:res:delta} Fig.~\ref{fig:delta} varies the reconfiguration delay, i.e., the ratio between circuit setup and transfer times. With microsecond-scale delays ($[0.01,0.1]$~ms), reconfiguration hardly matters: Greedy and Greedy-R are within 1.14 and 1.12 of $\TLP$, while LP-Greedy loses more (1.55) because its heavy-core restriction excludes cores that are good in this regime; LP-Greedy+ recovers 1.12. For $[0.1,1]$~ms, the splitting scheme LP-Split is the best (1.42), followed by LP-Greedy+ (1.51), LP-Greedy (1.56), and Greedy-R (1.58). As the delays grow, circuits become expensive and schemes that open many circuits deteriorate: for $[10,100]$~ms, Greedy and Weaver reach 10.25 and 10.33, and LP-Split and LP-Match 4.51 and 4.89, whereas LP-Greedy (2.62), LP-Max (2.79), Greedy-R (2.31), and LP-Greedy+ (2.30) remain stable.

\begin{figure}[!t]
    \centering
        \includegraphics[width=3.4in]{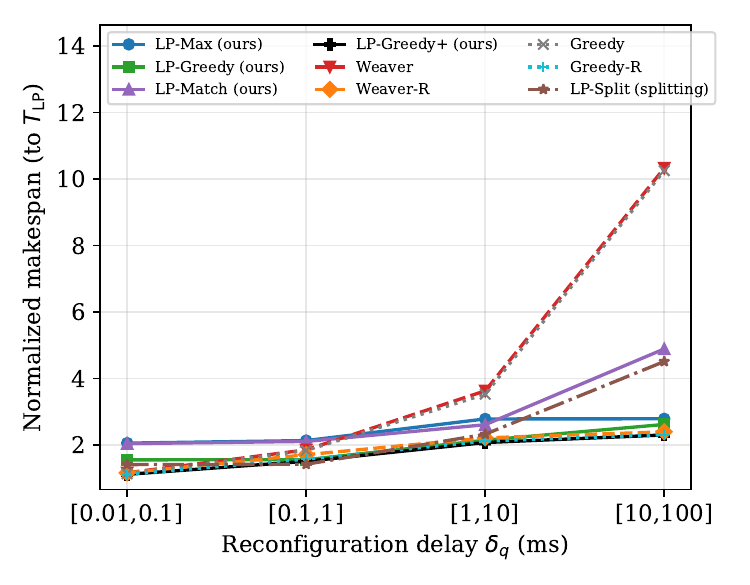}
    \caption{Normalized makespan versus the range of the reconfiguration delay $\delta_q$ ($m=10$).}
    \label{fig:delta}
\end{figure}

\textbf{Heterogeneous and skewed networks.} Table~\ref{tab:stress} considers two harder settings ($m=10$, 200 instances each). In the \emph{heterogeneous} setting, the rates of OCS cores are drawn from $\{10,\ldots,1000\}$~Gb/s, those of EPS cores from $\{1,\ldots,10\}$~Gb/s, and the delays from $\{1,\ldots,100\}$~ms, so that the cores differ by two orders of magnitude. Here, ignoring reconfiguration delays is disastrous: Greedy and Weaver obtain 12.71 and 13.10 on average and up to 23.61, whereas LP-Greedy obtains 2.40 (at most 3.61) without any tuning to the setting. Greedy-R (2.06) and LP-Greedy+ (2.04, at most 2.93) are the best schemes. In the \emph{skewed} settings, the endpoints of every flow are drawn from a Zipf distribution with exponent $1$ or $1.5$, so that a few ports carry most of the traffic. With exponent $1.5$, LP-Greedy (1.48) is clearly better than Greedy (2.05) and Weaver (2.05), close to Greedy-R (1.43), and LP-Greedy+ (1.39) matches the splitting scheme LP-Split (1.37). LP-Max suffers from skew (3.37), because hot ports have many flows that it assigns to the same core.

\begin{table}[!t]
\caption{Normalized makespan in heterogeneous and skewed networks ($m=10$; mean of $50$ instances; for the heterogeneous setting also the maximum).}
\centering
\setlength{\tabcolsep}{4.5pt}
\begin{tabular}{|l|c|c|c|c|}
\hline
\multirow{2}{*}{Scheme} & \multicolumn{2}{c|}{Heterogeneous} & Zipf $1$ & Zipf $1.5$ \\ \cline{2-5}
 & mean & max & mean & mean \\ \hline
LP-Max & 2.49 & 4.26 & 2.51 & 3.37 \\ \hline
LP-Greedy & 2.40 & 3.61 & 1.70 & 1.48 \\ \hline
LP-Match & 4.81 & 7.54 & 2.43 & 2.44 \\ \hline
LP-Greedy+ & 2.04 & 2.93 & 1.63 & 1.39 \\ \hline
Weaver & 13.10 & 25.68 & 2.39 & 2.05 \\ \hline
Weaver-R & 2.14 & 3.11 & 1.98 & 1.77 \\ \hline
Greedy & 12.71 & 23.61 & 2.36 & 2.05 \\ \hline
Greedy-R & 2.06 & 3.02 & 1.70 & 1.43 \\ \hline
LP-Split & 4.95 & 8.71 & 1.61 & 1.37 \\ \hline
\end{tabular}
\label{tab:stress}
\end{table}

\subsection{Comparison with Exact Optima}\label{sec:exact}
To separate the loss of the algorithms from the gap of the relaxation, we solve small instances exactly ($N=4$, $5$ coflows, about 46 flows, $m=3$). For every instance we compute, by a MILP, $\mathrm{IP}$, the minimum over all integral assignments of the largest post-rounding port busy time. By Lemma~\ref{lem:relax}, $\TLP\le\mathrm{IP}\le\OPT$, and for EPS cores $\mathrm{IP}=\OPT$ because BvN attains the largest port load. Table~\ref{tab:exact} reports the ratio of every scheme to $\mathrm{IP}$ for 200 instances with three EPS cores (where it is the ratio to $\OPT$) and 200 hybrid instances with one core of each type (where it is an upper bound on the ratio to $\OPT$). All MILPs were solved to optimality.

For the EPS instances, $\OPT/\TLP$ is 1.03 on average, so the relaxation is nearly tight; the LP-based roundings lose 1.28 to 1.41 times $\OPT$ on average, whereas the greedy rules, which coincide here because EPS cores have no delays, are within 3\% of the optimum. The hybrid instances are where reconfiguration matters. $\mathrm{IP}/\TLP$ is 1.07, and scheduling the IP-optimal assignment with the single-core schedulers gives 1.26 times $\mathrm{IP}$, which reflects the factors of the OCS schedulers rather than the assignment. Relative to $\mathrm{IP}$, Greedy and Weaver obtain 1.99 and 1.98 on average (up to 4.49 and 5.96), LP-Greedy 1.38 (up to 2.04), Greedy-R 1.36 (up to 2.19), and LP-Greedy+ 1.28 (up to 1.76), i.e., the combination is clearly better than each of its two components.

\begin{table}[!t]
\caption{Makespan divided by $\mathrm{IP}$ on small instances (mean and maximum of 100 instances; for EPS, $\mathrm{IP}=\OPT$).}
\centering
\begin{tabular}{|l|c|c|c|c|}
\hline
\multirow{2}{*}{Scheme} & \multicolumn{2}{c|}{3 EPS cores} & \multicolumn{2}{c|}{EPS + n.-a.-s. + a.-s.} \\ \cline{2-5}
 & mean & max & mean & max \\ \hline
LP-Max & 1.41 & 2.43 & 1.40 & 2.23 \\ \hline
LP-Greedy & 1.28 & 1.88 & 1.38 & 2.04 \\ \hline
LP-Match & 1.49 & 2.27 & 1.54 & 2.90 \\ \hline
LP-Greedy+ & 1.03 & 1.21 & 1.28 & 1.76 \\ \hline
Weaver & 1.05 & 1.48 & 1.98 & 5.96 \\ \hline
Weaver-R & 1.05 & 1.48 & 1.39 & 2.19 \\ \hline
Greedy & 1.03 & 1.21 & 1.99 & 4.49 \\ \hline
Greedy-R & 1.03 & 1.21 & 1.36 & 2.19 \\ \hline
\end{tabular}
\label{tab:exact}
\end{table}

\subsection{Worst-Case Behavior}\label{sec:worst}
Average results do not reveal how badly a scheme can perform. Since the true approximation ratio cannot be measured, we search for bad instances: starting from a random instance, a hill-climbing search repeatedly changes a flow size, an endpoint, a rate, or a delay and keeps the change if the normalized makespan of a target scheme does not decrease ($2000$ steps per run; flow sizes, rates, and delays stay below $10^4$). We run the search against every scheme on \emph{small} instances ($N=4$, $12$ flows, one core of each type; 20 runs per scheme) and \emph{larger} instances ($N=6$, $24$ flows, two cores of each type; 10 runs per scheme), and evaluate every scheme on all 180 instances found, relative to $\mathrm{IP}$ computed by a MILP (all solved to optimality). Table~\ref{tab:worst} reports the worst ratio of each scheme over these instances.

The data-only Greedy has no bounded ratio: the search drives it to ratios above $10^{6}$, limited only by the parameter range, by making reconfiguration delays large relative to flow sizes, as in Proposition~\ref{prop:unbounded}. The reconfiguration-aware heuristics reach 3.29 and 3.65 (Greedy-R) and 3.02 and 3.95 (Weaver-R), while LP-Greedy reaches 3.00 and 3.51 and LP-Greedy+ 2.58 and 2.94. The instances that are bad for one scheme are often benign for the other: on the instances found against Greedy-R, LP-Greedy+ obtains 1.27 on average versus 2.65 for Greedy-R, and on those found against LP-Greedy, Greedy-R obtains 1.43 versus 2.46 for LP-Greedy. These searches are heuristic and do not bound the ratios, and their outcome depends on the search; they show, however, that the guarantee of LP-Greedy is not paid for by a worse behavior on bad instances, and that the combination with a heuristic protects against the weaknesses of both.

\begin{table}[!t]
\caption{Worst ratio to $\mathrm{IP}$ over the instances found by adversarial search.}
\centering
\begin{tabular}{|l|c|c|}
\hline
Scheme & small ($N{=}4$, $m{=}3$) & larger ($N{=}6$, $m{=}6$) \\ \hline
Greedy & $>10^{6}$ & $>10^{6}$ \\ \hline
Weaver-R & 3.02 & 3.95 \\ \hline
Greedy-R & 3.29 & 3.65 \\ \hline
LP-Max & 4.36 & 2.93 \\ \hline
LP-Greedy & 3.00 & 3.51 \\ \hline
LP-Greedy+ & 2.58 & 2.94 \\ \hline
\end{tabular}
\label{tab:worst}
\end{table}

\subsection{Cost of Prohibiting Splitting}
For $m=10$, allowing splitting changes the LP bound only marginally: $\TLP/\Phi_{\mathrm{split}}$ is 1.0001 on average (and 1.022 for delays in $[10,100]$~ms). Hence, at the level of the relaxation, the non-splitting restriction costs almost nothing in these workloads. LP-Split, which follows an optimal splittable LP solution, splits 87\% of the flows (88\% of the data) into 2.5 parts on average, and the parts of a split flow complete on average 41.6~ms apart, which is roughly the time for which the receiver must buffer and resequence the data. Its makespan is 2.33 times $\Phi_{\mathrm{split}}$, whereas the best non-splitting scheme of each instance achieves 2.07 and LP-Greedy 2.15. Splitting does not help here because every part sent through an OCS core needs its own circuit. With small delays ($[0.1,1]$~ms), LP-Split is the best scheme in Fig.~\ref{fig:delta}, i.e., splitting pays off when circuits are cheap, while for millisecond-scale delays the non-splitting schemes are better.

\subsection{Total Weighted Coflow Completion Time}
We evaluate Algorithms~\ref{Alg7} and~\ref{Alg8} (LP-Group) on networks with $N=10$, $20$ coflows (about 463 flows), and $m=5$ cores (one EPS, two not-all-stop OCS, and two all-stop OCS cores). The weights are either drawn uniformly from $\{1,\ldots,10\}$ or all equal to $1$. Release times follow a Poisson process whose mean inter-arrival time is $\alpha\,\Phi_{\mathrm{split}}/20$, so the arrivals span about $\alpha$ times a makespan lower bound; $\alpha=0$ means zero release times. We use 200 instances per setting and normalize by the optimal value of LP~(\ref{coflow:main:wc}), which is at most $\OPT$ (Lemma~\ref{lem:wrelax}). The baselines process coflows one at a time on every core, in a priority order among the released coflows, and schedule each coflow by the single-core schedulers: \textbf{SEBF-Greedy} orders coflows by $w_k/\Gamma_k$, where $\Gamma_k$ is the bottleneck of coflow $k$ when it is alone (the smallest-effective-bottleneck-first rule of Varys~\cite{Chowdhury2014}, weighted), and assigns its flows by Greedy; \textbf{FIFO-Greedy} processes coflows in the order of release; \textbf{Weaver-SEBF} uses the SEBF order and Weaver's assignment. As for the makespan, these baselines ignore reconfiguration delays when assigning and ordering, and \textbf{SEBF-Greedy-R} and \textbf{FIFO-Greedy-R} are their reconfiguration-aware variants, which use Greedy-R and include the delays in $\Gamma_k$. \textbf{LP-Group+} returns the better of the schedules of LP-Group and SEBF-Greedy-R and keeps the guarantee of Theorem~\ref{thm:9}.

Fig.~\ref{fig:weighted} and Table~\ref{tab:weighted} show the results. LP-Group is better than every baseline that ignores reconfiguration delays in all six settings; with zero release times, it reduces the objective of SEBF-Greedy by 64\% (6.33 versus 17.64 for random weights, 6.15 versus 18.19 for equal weights) and is also better than SEBF-Greedy-R (7.97 and 8.66). With release times, SEBF-Greedy-R is better than LP-Group (e.g., 3.61 versus 4.74 for random weights and $\alpha=1$): it processes coflows one at a time as they arrive, whereas LP-Group starts a group only after all of its coflows have been released, as required by the analysis, which costs more when arrivals are spread out. LP-Group+ is the best scheme in every setting (e.g., 3.56 for $\alpha=1$) and keeps the guarantee. Ordering the flows inside each group by $C_k/w_k$ (Algorithm~\ref{Alg8}) reduces the objective of LP-Group by 6\% on average compared with an unordered group schedule. All normalized values are far below the guarantee of Theorem~\ref{thm:9}; they also include the gap of the interval-indexed relaxation, in which a completion time in $I_l$ is only bounded from below by $t_{l-1}=t_l/2$.

\begin{figure}[!t]
    \centering
        \includegraphics[width=3.4in]{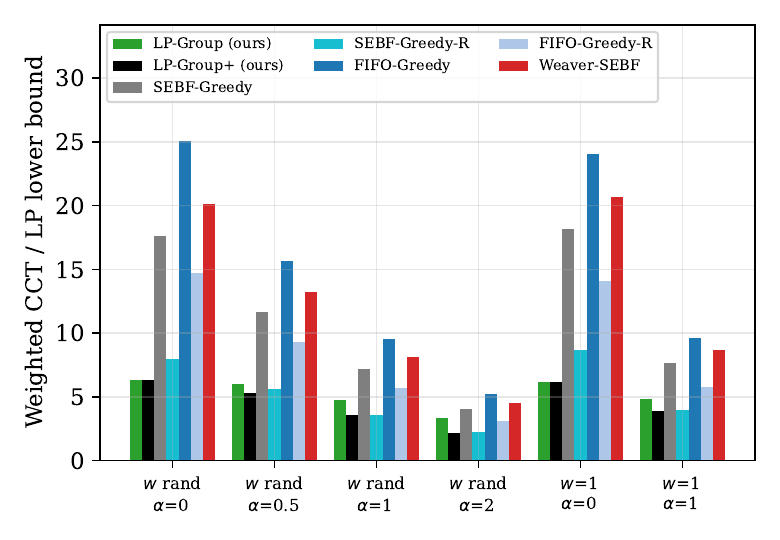}
    \caption{Total weighted coflow completion time divided by the optimal value of LP (\ref{coflow:main:wc}) ($m=5$, $20$ coflows).}
    \label{fig:weighted}
\end{figure}

\begin{table*}[!t]
\caption{Normalized total weighted completion time for random or equal weights and release spread $\alpha$.}
\centering
\small\setlength{\tabcolsep}{4pt}
\begin{tabular}{|l|c|c|c|c|c|c|c|}
\hline
Setting & \shortstack{LP-\\Group} & \shortstack{LP-\\Group+} & \shortstack{SEBF-\\Greedy} & \shortstack{SEBF-\\Greedy-R} & \shortstack{FIFO-\\Greedy} & \shortstack{FIFO-\\Greedy-R} & \shortstack{Weaver-\\SEBF} \\ \hline
rand., $\alpha{=}0$ & 6.33 & 6.31 & 17.64 & 7.97 & 25.05 & 14.72 & 20.09 \\ \hline
rand., $\alpha{=}0.5$ & 6.01 & 5.33 & 11.67 & 5.65 & 15.67 & 9.26 & 13.25 \\ \hline
rand., $\alpha{=}1$ & 4.74 & 3.56 & 7.19 & 3.61 & 9.57 & 5.68 & 8.15 \\ \hline
rand., $\alpha{=}2$ & 3.31 & 2.19 & 4.01 & 2.20 & 5.20 & 3.12 & 4.54 \\ \hline
equal, $\alpha{=}0$ & 6.15 & 6.15 & 18.19 & 8.66 & 24.04 & 14.07 & 20.64 \\ \hline
equal, $\alpha{=}1$ & 4.79 & 3.86 & 7.66 & 3.93 & 9.65 & 5.73 & 8.68 \\ \hline
\end{tabular}
\label{tab:weighted}
\end{table*}

\subsection{Running Time and Memory}
Table~\ref{tab:runtime} reports, for growing numbers of ports and coflows ($m=10$, 3 instances each), the time to compute $\TLP$ with Algorithm~\ref{Alg1} (including all LP solutions), the number of LP solutions, the largest number of LP variables, the times of LP-Greedy and LP-Match rounding and of the single-core scheduling of all cores, and the peak memory of the process. Because of the aggregation of Lemma~\ref{lem:aggregate}, the LP size depends on the number of port pairs and cores rather than on the number of flows, and only a few LP solutions are needed in practice because the search interval of Algorithm~\ref{Alg1} can start from $\Phi(S_s)$. For 16394 flows ($N=40$), computing $\TLP$ takes 14.1~s, the matching-based rounding 0.04~s, and scheduling 0.07~s. LP-Max needs at most 1~ms, and Greedy-R, which LP-Greedy+ runs in addition, a similar time as LP-Greedy. The peak memory includes about $80$~MB for the Python runtime and libraries.

\begin{table}[!t]
\caption{Running time (s) and peak memory (MB) for $m=10$.}
\centering
\setlength{\tabcolsep}{2.6pt}
\begin{tabular}{|c|c|c|c|c|c|c|c|c|c|}
\hline
$N$ & $K$ & flows & $\TLP$ & LPs & vars & LP-G & LP-M & sched & MB \\ \hline
10 & 100 & 2320 & 0.34 & 3 & 2784 & 0.01 & 0.01 & 0.00 & 80 \\ \hline
20 & 100 & 5348 & 2.33 & 3 & 11020 & 0.02 & 0.01 & 0.01 & 96 \\ \hline
30 & 100 & 9899 & 6.15 & 3 & 24766 & 0.05 & 0.02 & 0.03 & 121 \\ \hline
40 & 100 & 16394 & 14.07 & 3 & 42646 & 0.10 & 0.04 & 0.07 & 152 \\ \hline
10 & 200 & 4822 & 0.33 & 3 & 2874 & 0.02 & 0.01 & 0.00 & 81 \\ \hline
10 & 400 & 9425 & 0.33 & 3 & 2901 & 0.04 & 0.02 & 0.00 & 83 \\ \hline
\end{tabular}
\label{tab:runtime}
\end{table}

\textbf{Summary.} Against schemes that ignore reconfiguration delays, which include the existing non-splitting scheme for parallel networks (Weaver) and the data-only greedy baselines, the LP-based schemes are much better: LP-Greedy reduces the makespan by up to 54\% for $25$ cores and by a factor of more than $5$ in heterogeneous networks, and LP-Group reduces the weighted completion time by up to 64\%. Reconfiguration-aware heuristics that we designed for this evaluation (Greedy-R, Weaver-R) perform similarly to LP-Greedy on average, but they come with no guarantee, and their worst ratios found by the adversarial search are higher than those of LP-Greedy (by a small margin on the larger instances). The guarantee of our algorithms does not have to be traded for average performance: the combinations LP-Greedy+ and LP-Group+ keep the proven bounds, are the best non-splitting schemes in all experiments, and certify their quality on every instance through the LP bound.


\section{Concluding Remarks}\label{sec:Conclusion}
This paper studied non-splitting coflow scheduling in heterogeneous parallel networks with EPS, not-all-stop OCS, and all-stop OCS cores, where a circuit may carry several flows of its port pair. We gave an LP relaxation that charges reconfiguration delays per circuit, computed its makespan bound exactly in polynomial time, and analyzed a heavy-core rounding combined with established single-core schedulers. For the makespan, the proven bounds are $\min\{\tau, m\}$, $2\min\{\tau, m\}$, and $2\min\{2\tau-1, m+\tau-1, 2m\}$ for EPS, not-all-stop OCS, and all-stop OCS, respectively, and a hybrid network inherits the bound of its weakest architecture, which is at most $4m$. For the total weighted coflow completion time, an interval-indexed LP and an explicit group schedule that waits for release times yield $8R$ and $8R+4$ for zero and arbitrary release times. Heuristics that ignore reconfiguration delays have no bounded ratio, while our guarantees hold for every instance and survive the combination with any heuristic. The simulations show that the LP-based schemes stay far below their bounds, clearly outperform schemes that ignore reconfiguration delays, match reconfiguration-aware heuristics on average and are better on adversarially found instances, and that the combinations LP-Greedy+ and LP-Group+ are the best non-splitting schemes while keeping the guarantees; prohibiting splitting costs little unless circuits are cheap. The guarantees are upper bounds that grow with the number of cores; whether constant-factor guarantees independent of $m$ and $\tau$ can be achieved for the hybrid setting, and how to handle online arrivals and multi-stage collectives, are interesting directions for future work.

\section*{Acknowledgments}
This research was supported in part by the National Science and Technology Council, R.O.C. under Grant no. NSTC 114-2221-E-006-182 and NSTC 115-2221-E-006-196.

The author used Claude Opus~5.5 (Anthropic) as an AI assistant in preparing this article. It was used (i)~to edit the language and presentation of all sections and appendices; (ii)~to revise the presentation of the author's proofs in Sections~\ref{sec:Algorithm3} and~\ref{sec:Algorithm5}, to add some new statements and proofs, such as Lemmas~\ref{lem:relax} and~\ref{lem:bsearch}, and to draft the discussion of the guarantees, all of which the author checked; (iii)~to implement the simulation programs in Python and to generate the figures and tables of Section~\ref{sec:Results} and of the appendices from the simulation results; and (iv)~to check the bibliographic data of the references. The experiments were run by the author, and the author takes full responsibility for the content of this article.

\appendix
\section{Birkhoff-von Neumann Decomposition}\label{sec:BnV}
This section shows that the flows assigned to an EPS core can be transmitted within the largest port load of the core. Since circuits are not involved, the flows of the same port pair are aggregated, and we write $D=(d_{i,j})_{i,j=1}^{N}$ for the aggregated demand of the core measured in transmission time (data divided by the link rate). Let
\begin{equation*}
\rho(D)=\max \Big\{\max_{i}\sum_{j'=1}^{N}d_{i,j'},\ \max_{j}\sum_{i'=1}^{N}d_{i',j}\Big\}
\end{equation*}
be the load of $D$. The method decomposes $D$ into permutation matrices~\cite{marcus1959diagonals}, following Qiu \textit{et al.}~\cite{Qiu2015}.

\begin{thm}\label{BvN}
\cite{marcus1959diagonals} \textbf{(BvN theorem)} A nonnegative matrix $A\in \mathbb{R}^{N\times N}$ whose row and column sums are all equal to $\rho$ can be written as $A=\sum_{u=1}^{U} c_{u}\Pi_{u}$, where $c_{u}>0$, $\sum_{u=1}^{U} c_{u}=\rho$, every $\Pi_{u}$ is a permutation matrix, and $U\leq N^2-2N+2$.
\end{thm}

\begin{algorithm}
\caption{Birkhoff-von Neumann Decomposition~\cite{Qiu2015}}
    \begin{algorithmic}[1]
				\REQUIRE nonnegative demand matrix $D$ (in time units)
				\STATE $\rho \leftarrow \rho(D)$; $\widetilde{D}\leftarrow D$.
				\STATE $\vartriangleright$ Augmentation: make all row and column sums equal to $\rho$. \label{BvN_alg:1}
				\WHILE{some row sum or some column sum of $\widetilde{D}$ is smaller than $\rho$}
						\STATE Choose a row $i^*$ and a column $j^*$ whose sums are smaller than $\rho$.
						\STATE $p\leftarrow\min\Big\{\rho-\sum_{j'}\widetilde{d}_{i^*,j'},\ \rho-\sum_{i'}\widetilde{d}_{i',j^*}\Big\}$.
						\STATE $\widetilde{d}_{i^*,j^*}\leftarrow \widetilde{d}_{i^*,j^*}+p$.
				\ENDWHILE \label{BvN_alg:2}
				\STATE $\vartriangleright$ Decomposition and transmission. \label{BvN_alg:3}
				\STATE $t\leftarrow 0$.
				\WHILE{$\widetilde{D}\neq \mathbf{0}$}
						\STATE Let $G$ be the bipartite graph with an edge $(i,j)$ if and only if $\widetilde{d}_{i,j}>0$.
						\STATE Find a perfect matching $M$ in $G$; let $\Pi$ be its permutation matrix.
						\STATE $c\leftarrow\min\{\widetilde{d}_{i,j}\mid (i,j)\in M\}$.
						\STATE For every $(i,j)\in M$ with remaining demand, transmit the flows from $i$ to $j$ during $[t,t+c)$.
						\STATE $\widetilde{D}\leftarrow \widetilde{D} -c\,\Pi$; $t\leftarrow t+c$. \label{BvN_alg:4}
				\ENDWHILE
   \end{algorithmic}
\label{BvN_alg}
\end{algorithm}

In Algorithm~\ref{BvN_alg}, the augmentation amount $p$ is the smaller of the two deficits, so that no row or column sum exceeds $\rho$, and the decomposition loop runs while the augmented matrix is nonzero, so that an already balanced matrix is also decomposed.

\begin{thm}\label{BvN2}
Algorithm~\ref{BvN_alg} runs in polynomial time and transmits all demand of $D$ within $\rho(D)$ time. Hence the makespan of an EPS core equals its largest port load.
\end{thm}
\begin{proof}
Every iteration of the augmentation loop saturates a row or a column, so it ends after at most $2N-1$ iterations, and afterwards all row and column sums equal $\rho$ while $\widetilde{d}_{i,j}\ge d_{i,j}$. Throughout the decomposition loop, all row and column sums of $\widetilde{D}$ remain equal, because every step subtracts $c$ from exactly one entry of each row and of each column. Hence $\widetilde{D}/\rho'$, where $\rho'$ is the current common sum, is doubly stochastic, and by Theorem~\ref{BvN} (or Hall's theorem) the support of $\widetilde{D}$ contains a perfect matching. Every step makes at least one entry zero, so there are at most $N^2$ steps, each solvable by a bipartite matching algorithm. The common sum decreases by $c$ per step, so the total time is $\rho$. Finally, every entry $(i,j)$ receives transmission time $\widetilde{d}_{i,j}\ge d_{i,j}$, so all demand is transmitted. Since every schedule needs at least the largest port load, the schedule is optimal for the core.
\end{proof}

\section{Single-Core OCS Schedulers}\label{supp:sched}
In both OCS models, the flows of the same port pair assigned to a core are served through circuits of that pair (circuit reuse); let $D$ denote the aggregated demand of the core in time units, $L_{ij}=\delta+d_{i,j}$ the busy time of pair $(i,j)$ including one circuit establishment, and $T_i=\sum_{j: d_{i,j}>0}L_{ij}$ and $T_j=\sum_{i: d_{i,j}>0}L_{ij}$ the port busy times, i.e., $T^{r}_{iq}$ and $T^{r}_{jq}$ of the main paper.

\textbf{Not-all-stop OCS (Sunflow).} We implement the non-preemptive, work-conserving circuit policy of Sunflow~\cite{huang2016}: every pair with positive demand is served by one circuit, which is established once and then carries the whole aggregated demand of the pair; whenever an input port and an output port with pending demand between them are both idle, a circuit is established between them. Pending pairs are considered in non-increasing order of $L_{ij}$ (for the weighted algorithm, most urgent first).
\begin{lem}\label{lem:sunflow}
Every pair $(i,j)$ completes by $T_i+T_j-L_{ij}\le T_i+T_j$.
\end{lem}
\begin{proof}
Let $S$ be the time at which the circuit of pair $(i,j)$ starts to be established. Since the policy is work-conserving, at every time before $S$ port $i$ or port $j$ is busy with another pair. The busy time of port $i$ with other pairs is $T_i-L_{ij}$ and that of port $j$ is $T_j-L_{ij}$, so $S\le T_i+T_j-2L_{ij}$, and the pair completes at $S+L_{ij}$.
\end{proof}

\textbf{All-stop OCS (Reco-Sin).} Our implementation follows the regularization and stuffing steps of Reco-Sin~\cite{Tan2021}. Every nonzero entry of $D$ is rounded up to a multiple of $\delta$; the resulting matrix, measured in units of $\delta$, is an integer matrix whose maximum row or column sum we denote by $\rho'/\delta$. It is stuffed to a matrix whose row and column sums all equal $\rho'/\delta$ and decomposed into permutation matrices with positive integer multiplicities $c$, as in Algorithm~\ref{BvN_alg}. Every permutation matrix that serves real demand becomes one configuration, held as long as real demand remains but at most $c\delta$, and preceded by one reconfiguration of length $\delta$.
\begin{lem}\label{lem:recosin}
The schedule ends before $2(\rho(D)+\tau_q\delta)$, where $\tau_q$ is the maximum number of nonzero entries in a row or column of $D$.
\end{lem}
\begin{proof}
The multiplicities sum to $\rho'/\delta$, so there are at most $\rho'/\delta$ configurations, which require at most $\rho'$ reconfiguration time, and the configurations are held for at most $\rho'$ in total. Rounding increases every nonzero entry by less than $\delta$, so $\rho'<\rho(D)+\tau_q\delta$, and the length is at most $2\rho'<2(\rho(D)+\tau_q\delta)$.
\end{proof}

\section{Matching-Based Rounding (LP-Match)}\label{supp:match}
LP-Match adapts the rounding of Shmoys and Tardos~\cite{shmoys1993approximation}. It rounds the variables of each input port $i$ separately: for every core $q$, it creates $c_{iq}=\left\lceil \sum_{(i, j, k) \in \Fin{i}} x_{ijkq}\right\rceil$ slot nodes $u_{iq1},\ldots,u_{iqc_{iq}}$, distributes the fractional values of the flows of $\Fin{i}$, sorted by non-increasing size, over the slots, and computes a matching in the resulting bipartite graph $B_{i}(\textbf{x})$ between slot nodes and flow nodes. Algorithm~\ref{alg:match} gives the listing.

\begin{algorithm}
\caption{Matching-Based Rounding (LP-Match)}
    \begin{algorithmic}[1]
				\STATE Run Algorithm~\ref{Alg1} to obtain $\textbf{x}$; set $\hat{x}_{ijkq}=0$ for all $(i, j, k)$ and $q$.
				\FOR{each port $i\in \mathcal{I}$}
					\STATE Create $u_{iq1}, \ldots, u_{iqc_{iq}}$ for each $q\in \mathcal{M}$; set $x'(u_{iqr})=0$ for all $q$, $r$; $E_i\leftarrow\emptyset$.
					\FOR{each flow $(i, j, k)\in \Fin{i}$ in non-increasing order of $d_{ijk}$}
				    \FOR{each $q\in \mathcal{M}$ with $x_{ijkq}>0$}
								\STATE Find the minimum index $r$ such that $x'(u_{iqr})<1$.
								\IF{$x_{ijkq} \leq 1-x'(u_{iqr})$}
									\STATE Add the edge $(u_{iqr},(i,j,k))$ to $E_{i}$.
									\STATE $x'(u_{iqr}) \leftarrow x'(u_{iqr})+x_{ijkq}$.
								\ELSE
									\STATE Add the edges $(u_{iqr},(i,j,k))$ and $(u_{iq,r+1},(i,j,k))$ to $E_{i}$.
									\STATE $x'(u_{iq,r+1}) \leftarrow x_{ijkq}-(1-x'(u_{iqr}))$. \label{match:overflow}
									\STATE $x'(u_{iqr}) \leftarrow 1$. \label{match:sat}
								\ENDIF
						\ENDFOR
				\ENDFOR
				\STATE Find a matching $M_i$ in $B_{i}(\textbf{x})=(U_i\cup\Fin{i},E_i)$ that covers all flow nodes.
				\STATE For every edge $(u_{iqr},(i,j,k))\in M_i$, set $\hat{x}_{ijkq}=1$.
				\ENDFOR
				\STATE \textbf{return} $\hat{\textbf{x}}$.
   \end{algorithmic}
\label{alg:match}
\end{algorithm}

Fig.~\ref{fig:1} shows an example of the bipartite graph $B_i(\textbf{x})$ constructed from a fractional solution, and Fig.~\ref{fig:2} shows a matching that covers all flow nodes.

\begin{figure}[!ht]
    \centering
        \includegraphics[width=3.3in]{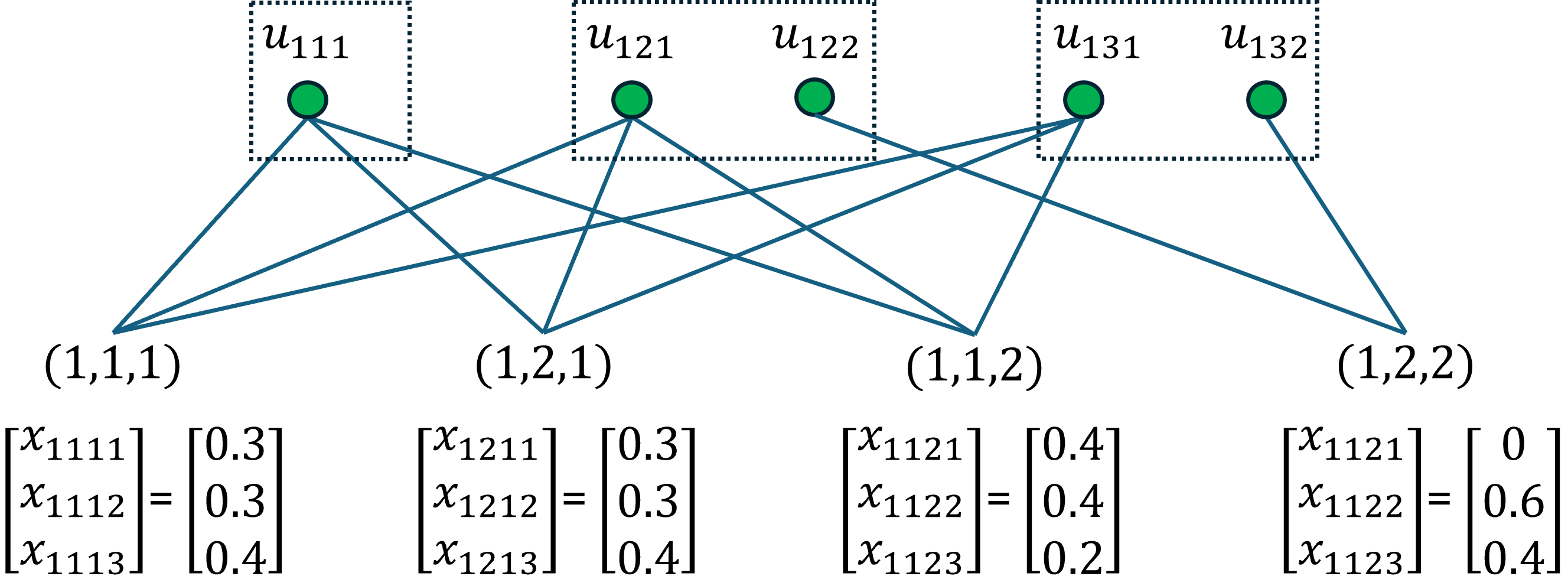}
    \caption{The bipartite graph $B_i(\textbf{x})$ constructed by LP-Match for an input port $i$.}
    \label{fig:1}
\end{figure}

\begin{figure}[!ht]
    \centering
        \includegraphics[width=3.3in]{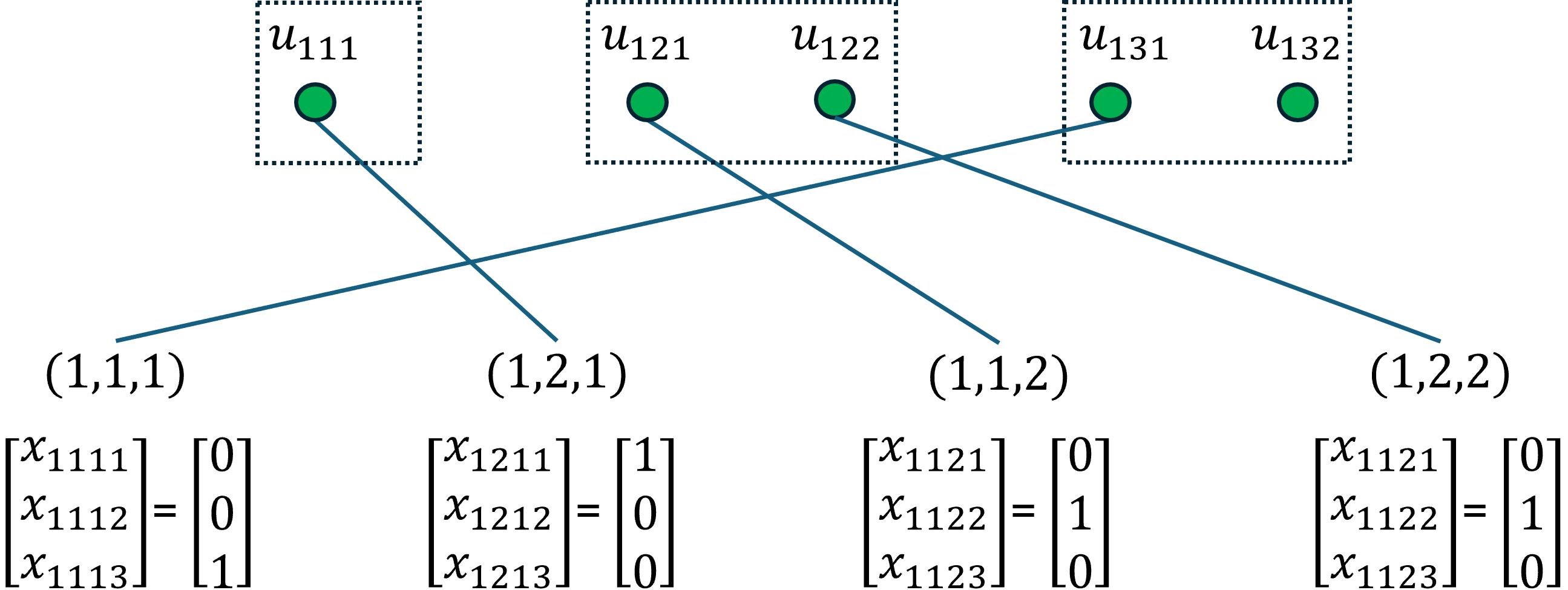}
    \caption{A matching $M_i$ in $B_i(\textbf{x})$ that covers all flow nodes.}
    \label{fig:2}
\end{figure}

\begin{lem}\label{lem:matchexist}
For every input port $i$, $B_i(\textbf{x})$ contains a matching that covers all flow nodes, and such a matching can be found in polynomial time.
\end{lem}
\begin{proof}
Assign to every edge the fraction that Algorithm~\ref{alg:match} allocates along it: $x_{ijkq}$ for the single edge of the first case, and $1-x'(u_{iqr})$ and $x_{ijkq}-(1-x'(u_{iqr}))$ for the two edges of the second case. Every flow node receives $\sum_q x_{ijkq}=1$ in total, and every slot node receives at most $1$, since its load never exceeds $1$; moreover, the $c_{iq}$ slots of core $q$ suffice because the total load of core $q$ is $\sum_{(i,j,k)\in\Fin{i}}x_{ijkq}\le c_{iq}$. Hence these fractions form a fractional matching of value $|\Fin{i}|$. Since the fractional matching polytope of a bipartite graph is integral, $B_i(\textbf{x})$ has an integral matching of size $|\Fin{i}|$, which covers all flow nodes; a maximum matching algorithm finds it.
\end{proof}

\begin{lem}\label{lem:support}
If $\hat{x}_{ijkq}=1$, then $x_{ijkq}>0$. If $x_{ijkq}=1$, then $\hat{x}_{ijkq}=1$.
\end{lem}
\begin{proof}
A flow node is connected to slots of core $q$ only if $x_{ijkq}>0$, and a flow with $x_{ijkq}=1$ is connected only to slots of core $q$.
\end{proof}

\textbf{Why both endpoints are controlled.} The matching is constructed for input ports only, but each flow receives a single assignment $\hat{x}_{ijkq}$, which is shared by its input port $i$ and its output port $j$. Lemma~\ref{lem:support} is a property of the assignment of each flow, not of the port at which the matching was computed. All bounds below use only this property and the LP constraints, which exist for input and output ports alike; hence they hold at both endpoints. The slot structure gives an additional bound at input ports only (Proposition~\ref{prop:perflow}(b)); at an output port $j$, the flows assigned to $q$ come from the matchings of different input ports, which is why the ratio of~\cite{shmoys1993approximation} cannot be replicated at both ends.

\begin{proof}[Proof of Proposition~\ref{prop:match}]
By Lemma~\ref{lem:support}, every flow assigned to core $q$ satisfies $x_{ijkq}>0$ and thus $p_{ijkq}\le\TLP$ by (\ref{interval:e}). As in the proof of Lemma~\ref{thm:5}, $T^{r}_{iq}\le\sum_{(i,j,k)\in\Fin{i}:\hat{x}_{ijkq}=1}p_{ijkq}\le\tau\TLP$, and symmetrically for output ports. The proof of Theorem~\ref{thm:6} with $\tau$ in place of $\min\{\tau,m\}$ gives $\tau$, $2\tau$, and $2\min\{\tau+\tau-1,2\tau\}=2(2\tau-1)$.
\end{proof}

\textbf{Per-flow reconfiguration.} A simpler convention charges one reconfiguration delay per flow, i.e., circuits are not reused. Under this convention, $\delta_q y_{ijq}$ in (\ref{interval:b}) and (\ref{interval:c}) is replaced by $\delta_q x_{ijkq}$ for every flow, so that (\ref{interval:b}) reads $\sum_{(i,j,k)\in\Fin{i}}p_{ijkq}x_{ijkq}\le T$, the variables $y$ disappear, and the single-core schedulers treat every flow as a separate entry. The results of the main paper continue to hold under this convention with the same proofs, where $\deg_q(i)$ is replaced by the number of flows of port $i$ on core $q$. In addition, LP-Match then satisfies the following sparsity-based bounds.
\begin{prop}\label{prop:perflow}
Under the per-flow convention, let $\textbf{x}$ be an extreme-point solution of LP($\TLP$). Then (a) $T^{r}_{iq},T^{r}_{jq}\le\min\{\tau,2Nm+1\}\TLP$ for all ports and cores, and (b) at every input port, $T^{r}_{iq}\le 2\TLP$.
\end{prop}
\begin{proof}
(a) An extreme point is determined by tight constraints among the $n$ assignment constraints and the $2Nm$ port constraints, so it has at most $n+2Nm$ nonzero variables~\cite{lenstra1990approximation, 10.5555/1965254}. Since a fractionally assigned flow has at least two nonzero variables, at most $2Nm$ flows are fractional. At port $i$ of core $q$, the integrally assigned flows on $q$ (which stay on $q$ by Lemma~\ref{lem:support}) contribute at most $\TLP$ by (\ref{interval:b}), and each of the at most $2Nm$ fractional flows contributes at most $\TLP$ by (\ref{interval:e}). Output ports are symmetric, and the bound $\tau\TLP$ is shown above.
(b) Consider core $q$ at input port $i$. The flows are processed in non-increasing order of $d_{ijk}$, hence of $p_{ijkq}$. Let $p^{\max}_r$ and $p^{\min}_r$ be the largest and smallest processing times on core $q$ of the flows adjacent to slot $u_{iqr}$. At most one flow is matched to each slot, so $T^{r}_{iq}\le\sum_r p^{\max}_r$. By the order, $p^{\max}_{r+1}\le p^{\min}_r$, and slot $r<c_{iq}$ is full, so $p^{\min}_r$ is at most the average of the processing times in slot $r$ weighted by the allocated fractions. Hence $T^{r}_{iq}\le p^{\max}_1+\sum_{r} p^{\min}_r\le \TLP+\sum_{(i,j,k)\in\Fin{i}}p_{ijkq}x_{ijkq}\le 2\TLP$.
\end{proof}

\section{Aggregated LP and a Lower Bound for Splitting}\label{supp:aggregate}
\begin{lem}\label{lem:aggregate}
Let $S$ be a set of allowed (flow, core) pairs. If two flows $f,f'$ of the same port pair are allowed on the same set of cores, then LP($T$) restricted to $S$ has an optimal solution with $x_{fq}=x_{f'q}$ for all $q$.
\end{lem}
\begin{proof}
Replace $x_{f\cdot}$ and $x_{f'\cdot}$ by $\bar{x}_{q}=(d_f x_{fq}+d_{f'}x_{f'q})/(d_f+d_{f'})$. The data terms of (\ref{interval:b}) and (\ref{interval:c}) are unchanged, (\ref{interval:a}) and the allowed set are preserved, and $\bar{x}_q\le\max\{x_{fq},x_{f'q}\}\le y_{ijq}$ preserves (\ref{interval:y}).
\end{proof}
\begin{lem}\label{lem:splitlb}
Let $\Phi_{\mathrm{split}}=\Phi(S_s)$ be the optimal value of LP($T$) without constraint (\ref{interval:e}). Every schedule in which flows may be split over several cores has makespan at least $\Phi_{\mathrm{split}}$.
\end{lem}
\begin{proof}
Take such a schedule with makespan $C$. Let $x_{ijkq}$ be the fraction of the data of flow $(i,j,k)$ sent through core $q$, and let $y_{ijq}=1$ if some data of port pair $(i,j)$ is sent through OCS core $q$. Then (\ref{interval:a}), (\ref{interval:y}), and (\ref{interval:d}) hold, and the argument in the proof of Lemma~\ref{lem:relax} shows that the busy time of every port of every core, which is at most $C$, is at least the left-hand side of (\ref{interval:b}) or (\ref{interval:c}) with $T=C$. Hence $(\textbf{x},\textbf{y},C)$ is feasible for the LP that defines $\Phi(S_s)$.
\end{proof}

By Lemma~\ref{lem:aggregate}, $\Phi(S)$ is computed by an LP with one variable per core and per class of flows that share a port pair and an allowed set. Since the allowed sets of the flows of a pair are nested (a core allows a flow if and only if its size is at most a core-specific threshold), every pair has at most $m+1$ classes, so the LP has $O(N^2m^2)$ variables, independently of the number of flows. All rounding variants are applied to the resulting optimal solution, which is feasible for LP($\TLP$).

\section{Additional Experiments}\label{supp:exp}
This section reports additional makespan experiments with the default parameters of Section~\ref{sec:Results} except for the varied factor; every point averages $100$ instances. The schemes are those of Section~\ref{sec:Results}.

\textbf{Workload distributions.} Table~\ref{tableProportion} lists five mixtures of the four coflow profiles: dense (Index 1), average (Index 2), large flows (Index 3), small flows (Index 4), and the default (Index 5). Fig.~\ref{fig:ratio5} shows the normalized makespan. For Indices 1, 2, 3, and 5, LP-Greedy and Greedy-R are within about $3\%$ of each other, LP-Greedy+ is the best non-splitting scheme, and Greedy and Weaver are $25$--$70\%$ worse than LP-Greedy. With small flows only (Index 4), circuit establishments dominate the transfer times: schemes that ignore reconfiguration delays or open many circuits deteriorate sharply (Greedy 14.81, Weaver 14.96, LP-Split 5.73, LP-Match 4.64), whereas LP-Greedy (2.11), LP-Max (2.08), Weaver-R (1.90), Greedy-R (1.82), and LP-Greedy+ (1.81) remain stable.

\begin{table}[!ht]
\caption{Mixtures of the profiles $(1, 5, 1, 10)$, $(1, 5, 10, 1000)$, $(5, N, 1, 10)$, and $(5, N, 10, 1000)$.}
\centering
\begin{tabular}{|c|c|}
\hline
 Index    & Workload Distribution \\ \hline
  1       &  $0\%$, $0\%$, $50\%$, $50\%$     \\ \hline
  2       &  $25\%$, $25\%$, $25\%$, $25\%$     \\ \hline
  3       &  $0\%$, $50\%$, $0\%$, $50\%$     \\ \hline
  4       &  $50\%$, $0\%$, $50\%$, $0\%$     \\ \hline
  5       &  $41\%$, $29\%$, $9\%$, $21\%$     \\ \hline
\end{tabular}
\label{tableProportion}
\end{table}

\begin{figure}[!ht]
    \centering
        \includegraphics[width=3.3in]{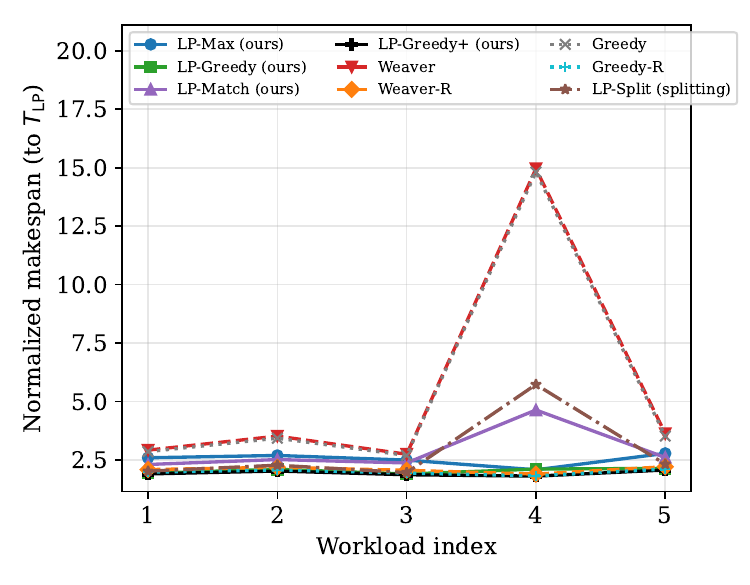}
    \caption{Normalized makespan for the workload distributions of Table~\ref{tableProportion}.}
    \label{fig:ratio5}
\end{figure}

\textbf{Number of coflows.} Fig.~\ref{fig:ratio6} varies the number of coflows from $120$ to $200$. The normalized makespans change little and decrease slightly as the load grows; for $200$ coflows, LP-Greedy achieves 1.95, Greedy-R 1.97, LP-Greedy+ 1.90, and Greedy 2.82.

\begin{figure}[!ht]
    \centering
        \includegraphics[width=3.3in]{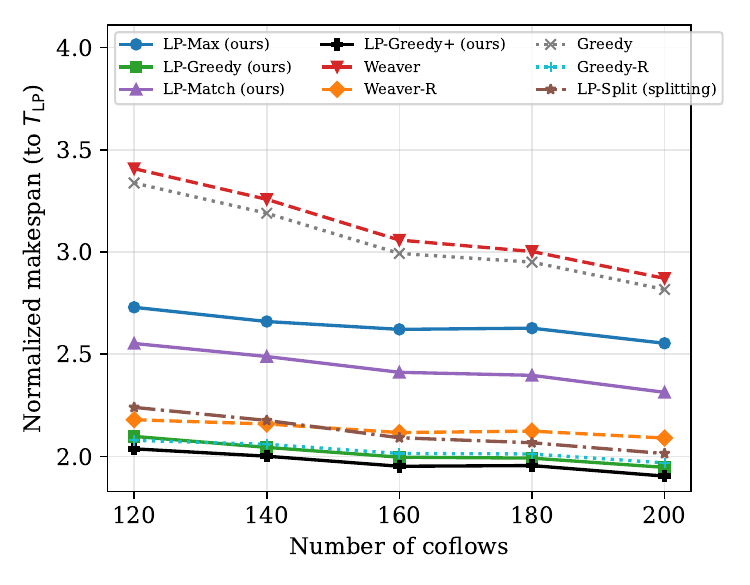}
    \caption{Normalized makespan for $120$ to $200$ coflows.}
    \label{fig:ratio6}
\end{figure}

\textbf{Types of optical cores.} Fig.~\ref{fig:ratio7} varies the fraction $\phi$ of not-all-stop OCS cores; the fraction of all-stop OCS cores is $0.8-\phi$. As $\phi$ grows, the normalized makespan of all schemes decreases, since not-all-stop cores do not suffer from the global coupling of reconfigurations; for $\phi=0.8$, where no all-stop core remains, LP-Greedy achieves 1.64, Greedy-R 1.58, LP-Greedy+ 1.56, and Greedy 2.82.

\begin{figure}[!ht]
    \centering
        \includegraphics[width=3.3in]{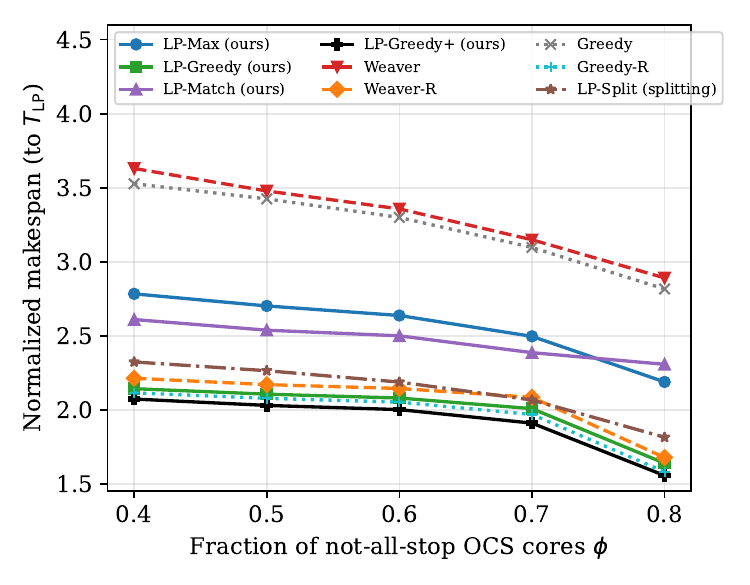}
    \caption{Normalized makespan for various fractions $\phi$ of not-all-stop OCS cores.}
    \label{fig:ratio7}
\end{figure}


\end{document}